\documentclass[reqno]{amsart}
\usepackage[dvipsnames]{xcolor}
\usepackage{amsmath}
\usepackage{graphicx, color}
\usepackage{chngcntr}
\usepackage{tikz}
\usetikzlibrary{patterns}
\usetikzlibrary{backgrounds}
\usepackage{amsmath,amsfonts, amssymb,amsfonts,amsthm}
\numberwithin{equation}{section}
\usepackage{wrapfig}
\usepackage{chngcntr}
\usepackage{apptools}
\usepackage{paralist}
\usepackage{graphics} 
 \usepackage{epsfig} 
 \usepackage[colorlinks=true]{hyperref}
 \usepackage{graphicx}  \usepackage{epstopdf}
 \usepackage{color}
 
 \usepackage{float}

 \usepackage[toc,page]{appendix}

\usepackage{tikz}
\usepackage{pgfplots}
\usepackage{amsmath,amssymb,bm}
\usepgfplotslibrary{fillbetween}
\usetikzlibrary{patterns}
\usetikzlibrary{scopes}
\usetikzlibrary{shapes,arrows}
\usetikzlibrary{plotmarks,intersections}
\usetikzlibrary{calc}
 
 \usepackage{bbm}
\hypersetup{urlcolor=blue, citecolor=red}

\def\currentvolume{}
 \def\currentissue{}
  \def\currentyear{}
   \def\currentmonth{}
    \def\ppages{X--XX}
     \def\DOI{}

\newtheorem{theorem}{Theorem}[section]

\newtheorem{lemma}[theorem]{Lemma}
\newtheorem{proposition}{Proposition}

\theoremstyle{definition}
\newtheorem{definition}[theorem]{Definition}
\newtheorem{remark}{Remark}

\usepackage{multirow}
\usepackage{array}
\usepackage{caption}
\usepackage{subcaption}
\title[]
   {Polytropic gas modelling at   kinetic and macroscopic levels}

\author[Vladimir Djordji\'c, Milana Pavi\'c-\v Coli\'{c} and Nikola Spasojevi\'c]{}

 \email{Djordjic@acom.rwth-aachen.de}
 \email{milana.pavic@dmi.uns.ac.rs}
 \email{nikola@utexas.edu}

\newcommand{\changelocaltocdepth}[1]{%
	\addtocontents{toc}{\protect\setcounter{tocdepth}{#1}}%
	\setcounter{tocdepth}{#1}%
}

\newcommand{\nocontentsline}[3]{}
\newcommand{\tocless}[2]{\bgroup\let\addcontentsline=\nocontentsline#1{#2}\egroup}

\begin{document}
\maketitle


\centerline{\scshape Vladimir Djordji\'c}

\medskip

{\footnotesize
	\centerline{Applied and Computational Mathematics}
	\centerline{RWTH Aachen University}
	\centerline{Schinkelstr. 2, 52062 Aachen, Germany}
}

\medskip

{\footnotesize
  \centerline{Department of Mathematics and Informatics}
  \centerline{Faculty of Sciences, University of Novi Sad}
 \centerline{Trg Dositeja Obradovi\'ca 4, 21000 Novi Sad, Serbia}
}

\medskip

\centerline{\scshape 
	Milana Pavi\'c-\v Coli\'{c}}

\medskip

{\footnotesize
	\centerline{Department of Mathematics and Informatics}
	\centerline{Faculty of Sciences, University of Novi Sad}
	\centerline{Trg Dositeja Obradovi\'ca 4, 21000 Novi Sad, Serbia}
}

\medskip

\centerline{\scshape
	Nikola Spasojevi\'c}

\medskip

{\footnotesize
	\centerline{Oden Institute for Computational Engineering and Sciences}
	\centerline{University of Texas at Austin} 
	\centerline{204 E 24th St, Austin TX 78712, USA }
}
\bigskip


\bigskip


\begin{abstract}
In this paper, we  consider the  kinetic model  of continuous type describing a polyatomic gas in two different settings corresponding to a different choice of the functional space used to define macroscopic quantities. Such a model introduces a single continuous variable supposed to capture all  the phenomena related to the more complex structure of a  molecule having more than one atom, such as internal degrees of freedom in a collision. In particular, we provide a direct comparison of these two settings, and show their equivalence  after  the distribution function is rescaled and    the cross section is reformulated. We then focus on the kinetic model for which the rigorous existence and uniqueness result in the space homogeneous case is recently proven.  Using the cross section proposed in that analysis together with the maximum entropy principle, we establish macroscopic models of six and fourteen fields. In the case of six moments, we calculate the exact, nonlinear, production term  and prove its total agreement with extended thermodynamics, as it satisfies the entropy residual inequality on the whole range of model validity. Moreover, for the fourteen moments model, we provide new expressions for relaxation times and transport coefficients in a linearized setting, that yield both matching with the experimental data for  dependence of the shear viscosity upon temperature and a satisfactory  agreement with the theoretical value   of the Prandtl number, on the room temperature range when only translational and rotational modes of molecules are taken into account, as much as on higher temperatures when vibrational modes appear as well.
\end{abstract}

\medskip





\section{Introduction}

This paper is devoted to  both the kinetic and macroscopic modelling of a polyatomic gas. In these models the core mechanism of particle interactions are  molecular  collisions. During a collision of monatomic gas molecules, only the translational  degrees of freedom occur, yielding that only the kinetic energy appear in the energy conservation law during a collision, and that the state of a gas can be described by the velocity distribution function that solves the Boltzmann equation \cite{Cerc, Sone 1, Sone 2}. The collisions of polyatomic molecules are much more intricate than the ones in  monatomic case, whether elastic or non-elastic, because of the  presence  of internal degrees of freedom -- except classical translation in the physical space ($\mathbb{R}^3$),  rotation or vibration of a polyatomic molecule can occur during the collision process. At the kinetic level, this is reflected on the microscopic energy conservation law during a molecular collision, where apart the usual kinetic energy of molecules it appears microscopic  internal energy as well. On the other hand, for  monatomic gases  all macroscopic quantities are identified as moments of the velocity distribution function, that satisfy  one single hierarchy of the balance laws in the which the flux in one equation becomes the density in the next one \cite{Mul-Rugg}. However, for polyatomic gases, 
the trace of momentum flux is not related to the gas internal energy density anymore, that resulted into two types of moment equations  within extended thermodynamics, namely \emph{momentum} and \emph{energy} like hierarchies \cite{MPC-Rugg-Sim, Rugg-Poly}.  \\

In the context of kinetic theory, the main difference in polyatomic gas modelling comes with the parametrization of microscopic energy collision law.

In the semi-classical approach \cite{W-C, Della, Magin C, Cha-Cow, Gio, Kusto-book, Groppi-Spiga}, only the molecular velocity is parameterized, while the internal energy of molecules takes  discrete values. Moreover, one distribution function is assigned  to each energy level, leading to the system of kinetic equations.  
On the other hand, continuous kinetic models \cite{LD-Bourgat,LD-Toulouse,DesMonSalv} introduce a single continuous variable  which sees internal degrees of freedom as a communicable   internal energy during  collisions. Then both molecular velocities \emph{and} molecular internal energies are parametrized: velocities are obtained by  introducing the classical scattering direction which splits the \emph{pure}   kinetic energy of the colliding particles, while for microscopic internal energy an additional parameter is introduced in order to distribute the proportion of \emph{ pure} internal energy to  each interacting molecule, following the ideas of Borgnakke-Larsen  procedure \cite{Bor-Larsen}.\\ 

The continuous kinetic models incorporate microscopic internal energy to the list of arguments of the distribution function,  which allows to  write a single Boltzmann equation describing a polyatomic gas. The collision operator has the  two key elements that are subject to the modelling: $(i)$ the cross section which encodes microscopic interaction law, and $(ii)$ the weight  function  that aims at recovering a proper energy law at the macroscopic level. Continuous  models precisely differ in the use of functional space as an environment where physical intuition is achieved:  for \cite{LD-Bourgat} physical quantities associated  to the kinetic model (such as gas  density, mean velocity, energy, etc.) are obtained by means of the plain $L^1$ space, and so we refer to this model as the model in the \emph{non-weighted} setting, while in \cite{LD-Toulouse,DesMonSalv}  weighted $L^1$ space arises and so we call this setting \emph{weighted}. Both settings are accurate in the case of \emph{polytropic} (or calorically perfect) gases, when the macroscopic internal energy of the gas is linear with respect to the gas temperature. For \emph{non-polytropic} (or thermally perfect) gases, when this dependence is nonlinear,  kinetic and macroscopic models are rewritten starting from the weighted setting \cite{Rugg-Bisi-6, Rugg-non-poly-14}, but   the weight function depending solely on the microscopic internal energy  remains unknown, and so models are still incomplete.\\

In this paper we restrict to polytropic gases, and perform a direct comparison of these continuous  models in two different settings. We show that they are equivalent, but only after the distribution function is appropriately  rescaled with the weight function \emph{and} the cross section is reformulated. Redefinition of  the cross section,  firstly pointed out in \cite{LD-Bisi}, Remark 1,  removes the  singularity in the collision operator strong form of  the model in the weighted setting, which opens the door to the mathematically rigorous theory, as, for instance, to the existence and uniqueness result   in the case of space homogeneity  \cite{MPC-IG-poly}, followed by the study of  polynomial and exponential moments.  \\

Another aspect of this paper is to build the moment equations starting from the continuous kinetic model describing a polyatomic gas in the non-weighted setting. In the monatomic case, they can be derived by three different approaches:  the Grad’s method \cite{Grad},  the maximum entropy principle \cite{Kogan, Drey, Leve}, and universal principles of  extended thermodynamics \cite{Mul-Rugg}. In the polyatomic case, the continuous  kinetic model in the weighted setting is extensively used as a basis for deriving  macroscopic models  starting from the kinetic theory \cite{Rugg-Poly}. 

This research path starts with the fourteen moments model  firstly introduced in \cite{MPC-Rugg-Sim}, and improved in many ways afterwards \cite{MPC-Sim, Rugg-Arima-Metrelli, MPC-Madj-Sim-Parma}. Kinetic theory provides an insight by calculating the production terms that allows for explicit expressions of relaxation times, which are of phenomenological nature in the macroscopic theories. For instance, this model is suitable for gases with large bulk viscosity \cite{TARS-E, Aoki,  Aoki-2, Str-Ra}.

On the other hand, when shear stresses and heat conduction are neglected, six fields model arise. In this model, the dominant non-equilibrium effect is the dynamic pressure, which is an excess normal pressure added to the standard thermodynamic pressure. The physical motivation for such a study is the fact that the bulk viscosity, and consequently the relaxation time for dynamic pressure is several order of magnitudes greater than the shear viscosity and heat conductivity \cite{Rugger-MEP, TARS-Acta}.   This model is of particular interest, since it is one of the rare systems that admits a non-linear closure of the governing equations using the entropy principle \cite{Rugg-Arima-Sugiyama-6-fields, Rugg-Arima-Sugiyama-6-fields-2, Rugg-Overshoot}.  It also admits the exact solution of the variational problem of maximum entropy principle, as shown in \cite{Rugger-MEP, Rugg-Bisi-6, MPC-Madj-Sim}. In particular, in \cite{MPC-Madj-Sim} the dynamic pressure is not introduced \emph{a priori}, but rather regarded as a measure  of deviation of the system from an equilibrium state, through the analysis  of pressure tensor trace.  The source term is calculated in  \cite{Rugg-Bisi-6} for the discrete energy model, and in  \cite{MPC-Madj-Sim}  for the continuous model in the weighted setting which can be related to the source term  of extended thermodynamics described in \cite{Rugg-Arima-Sugiyama-6-fields} for the  cross section which it is not Galilean invariant, but yields the Galilean invariant production term.\\

In this paper we take an another path by starting with the continuous model in the  non-weighted  setting 
and  build both six and fourteen moments models using the maximum entropy principle.  We are motivated by recent rigorous results from  \cite{MPC-IG-poly}, where a new model for the cross section is proposed. In this paper, we first provide its physical insight by computing the corresponding collision frequency in the equilibrium state. Then we show that it yields production terms which are in total agreement with the macroscopic theory of extended thermodynamics. More precisely, for the six fields model we prove that the residual inequality from  \cite{Rugg-Arima-Sugiyama-6-fields} is satisfied on the whole range of model validity. On the other hand, production terms for the fourteen moments model lead to  new expressions for relaxation times and transport coefficients. Since the cross section contains one parameter,   we will show that this parameter can be adjusted so that  the shear viscosity dependence upon temperature matches with the experimental data given in \cite{Cha-Cow} for room temperature range and in \cite{Exper-2, Exper-1} for high temperatures, and at the same time  recovers a  value of the Prandtl number that coincides at a satisfactory level with its theoretical estimate obtained by means of Eucken's relation. \\

The paper is organized as follows. Section \ref{Sec: Coll} studies collisions in polyatomic gases, and introduces the main notions of the continuous kinetic model. Then in Section \ref{Sec: kin mod} we describe  non-weighted and weighed settings for this model, that are further compared in Section \ref{Sec: Comp}. For the non-weighted setting we establish macroscopic models, namely the six fields model in Section \ref{Sec: macro six} and the fourteen moments model in Section \ref{Sec: macro 14}. The Appendix contains computations  of  the collision frequency and production terms for both macroscopic models.

\section{Study of  collisions in a polyatomic gas}\label{Sec: Coll}

In this section, we describe a collision process and introduce the main notions  used in continuous kinetic models \cite{LD-Bourgat, DesMonSalv, MPC-IG-poly}. \\

We assume that interactions between particles are binary collisions of polyatomic molecules. Due to the  complex structure of a polyatomic molecule, we need to take into account  internal degrees of freedom, as apart from the usual translation, there is a  possible rotation and vibration of molecules during the collision process.   The idea of \emph{continuous} kinetic models  is to capture these phenomena with a unique \emph{continuous}  variable $I$, that we call \emph{microscopic internal energy} of the molecule. In continuous models, a single Boltzmann equation governs evolution of the distribution function, that now has extended list of microscopic arguments -- besides the usual molecular velocity  $v\in \mathbb{R}^3$, it depends also on the  microscopic internal energy $I \in [0,\infty)$.\\

In order to study a collision process, we attribute the velocity-internal energy pair $(v, I)$ to each molecule. Then we consider the two colliding molecules, both of the same mass $m$, with pre-collisional molecular velocities and microscopic internal energies $(v',I')$ and $(v'_*, I'_*)$. After the collision, these quantities  transform to $(v, I)$ and $(v_*, I_*)$ respectively. Here we consider elastic collisions, meaning that  the total (kinetic+microscopic internal) energy of the molecular pair is conserved, and thus  conservation laws of momentum and energy  hold during the collision process,
\begin{equation}\label{micro CL}
\begin{split}
v +  v_* &=  v' + v'_*,\\
\frac{m}{2} \left|v\right|^2 + I + \frac{m}{2} \left|v_*\right|^2 + I_* &= \frac{m}{2} \left|v'\right|^2 + I' + \frac{m}{2} \left|v'_*\right|^2  + I'_*.
\end{split}
\end{equation}
These equations can be written in the reference frame of center-of-mass, by introducing velocity of the center of mass $V$ and relative velocity $u$,
\begin{equation}\label{cm-rv}
V:= \frac{v+v_*}{2}, \quad u:=v-v_*.
\end{equation}
Then \eqref{micro CL} can be rewritten,
\begin{equation}\label{micro CL cm}
\begin{split}
V &=  V',\\
\frac{m}{4} \left|u\right|^2 + I + I_* &= \frac{m}{4} \left|u'\right|^2 + I' + I'_* =: E.
\end{split}
\end{equation}
In order to describe the complete collision transformation, the aim is to express all pre-collisional quantities in terms of post-collisional ones. To that end, we use Borgnakke-Larsen procedure \cite{Bor-Larsen} that first introduces the parameter $R\in[0,1]$ in order to separate the pre-collisional kinetic energy $\frac{m}{4} \left|u'\right|^2 $ and the total microscopic internal energy $I' + I'_*$,
\begin{equation}\label{R}
\frac{m}{4} \left|u'\right|^2 = R E, \quad I'+I'_*=(1-R) E.
\end{equation}
 Then, the parameter $r\in[0,1]$ distributes the  total microscopic internal energy among the two colliding molecules, which implies
\begin{equation}\label{I' I'_*}
 I'=r(1-R) E, \quad I'_*=(1-r)(1-R) E.
\end{equation}
Finally, we parametrize the relative speed from \eqref{R} with a unit vector $\sigma \in S^2$, which yields expression for pre-collisional velocities using conservation of momentum \eqref{micro CL cm},
\begin{equation}\label{velocities}
v' = V + \sqrt{\frac{R E}{m}} \sigma, \quad v'_* = V - \sqrt{\frac{R E}{m}}\sigma.
\end{equation}
Relations \eqref{I' I'_*}--\eqref{velocities} together   with
\begin{equation}\label{coll mapping-2}
r'=\frac{I}{I+I_*} = \frac{I}{E-\frac{m}{4}\left| u \right|^2}, \quad R'=\frac{m \left|u\right|^2}{4 E}, \quad \sigma'=\frac{u}{\left|u\right|},
\end{equation}
define the   collision transformation
\begin{equation}\label{coll mapping}
T: (v, v_*, I, I_*, r, R, \sigma) \mapsto (v', v'_*, I', I'_*, r', R', \sigma').
\end{equation}
The Jacobian of this transformation  \cite{DesMonSalv,MPC-IG-poly} is computed in Lemma \ref{Lemma coll mapping}.
\begin{lemma}\label{Lemma coll mapping}
	The Jacobian of transformation $T$ given in \eqref{coll mapping}
	is given by
	\begin{equation}\label{coll mapping-3}
	J_T = \frac{(1-R) R^{\frac{1}{2}}}{(1-R')R'^{\frac{1}{2}}}= \frac{(1-R) \left|u'\right|}{(1-R')\left|u\right|}.
	\end{equation}
\end{lemma}

Using relations \eqref{I' I'_*}  and \eqref{coll mapping-2}   the invariance property of the function that involves the product $I I_* $ can be proven.
\begin{lemma}\label{Lemma fun r R}
 The following invariance holds
	\begin{equation*}
	I I_* \, r\,(1-r) \, (1-R)^2  = 	I' I'_* \, r'\,(1-r') \, (1-R')^2,
	\end{equation*}
	where the involved quantities are linked via the mapping \eqref{coll mapping}.
\end{lemma}
The proof of this lemma can be found in \cite{MPC-IG-poly}.

\section{Kinetic model for a polyatomic gas in two different settings}\label{Sec: kin mod}

In this section we introduce the continuous kinetic model in two different settings describing  a polyatomic gas, originating from references \cite{DesMonSalv, LD-Bourgat, LD-Toulouse}.  The goal is to write the Boltzmann equation governing the distribution function that probabilistically describes the state of a polyatomic gas.  

In both settings,  the distribution function depends on the usual macroscopic variables:  time $t\geq 0$ and space position $x \in \mathbb{R}^3$, but also on extended list of microscopic variables: molecular velocity $v\in \mathbb{R}^3$ and microscopic internal energy $I \in [0,\infty)$, i.e.
\begin{equation*}
f := f(t, x, v, I) \geq 0,
\end{equation*}
and is non-negative. Its measure of change, collision operator, acts only on microscopic variables $v$ and $I$. We have the two different definitions of collision operators, depending on the functional space we work in.\\

We define the plain $L^1$ space,
\begin{equation}\label{space L^1 plain}
\begin{split}
L^1 &= \left\{ f \ \text{measurable}: \int_{\mathbb{R}^3 \times [0,\infty)  } \left| f(v, I)  \right| \mathrm{d} I \mathrm{d}v < \infty \right\},
\end{split}
\end{equation}
and  the  $L^1$ space weighted with the suitable function $\varphi(I) \geq 0$,
\begin{equation}\label{space L^1 weight}
\begin{split}
L^1_\varphi &= \left\{ f \ \text{measurable}: \int_{\mathbb{R}^3 \times [0,\infty)  } \left| f(v, I)  \right| \varphi(I) \mathrm{d} I \mathrm{d}v < \infty \right\}.
\end{split}
\end{equation}
We describe the two settings below. 

\subsection{Kinetic model in the  non-weighted setting}\label{Sec: n-w model}
The kinetic model in the non-weighted setting is introduced in \cite{LD-Bourgat}. For the distribution function 
\begin{equation*}
f := f(t, x, v, I) \geq 0,
\end{equation*}
we write the Boltzmann equation
\begin{equation}\label{BE non-weight}
\partial_t f + v \cdot \nabla_{x} f = Q^{nw}(f,f)(v,I),
\end{equation}
with the collision operator $Q^{nw}(f,f)$  in the strong form defined below.
\subsubsection{Collision operator $Q^{nw}$ in the strong form}  We first introduce functions
\begin{equation}\label{fun r R}
\phi_\alpha(r) := (r(1-r))^{\alpha}, \qquad \psi_\alpha(R) :=  (1-R)^{2\alpha}.
\end{equation}
The strong form of collision operator in the non-weighted setting reads
\begin{multline}\label{Q non-weight}
Q^{nw}(f,f)(v,I) = \int_{\mathbb{R}^3 \times [0,\infty)\times [0,1]^2 \times S^2} \left( f' f'_* \left(\frac{I \, I_*}{I' \, I'_*} \right)^{\alpha}- f f_*\right)
\\
\times \mathcal{B}^{nw} (1-R) R^{\frac{1}{2}} \phi_\alpha(r) \, \psi_\alpha(R) \, \mathrm{d} \sigma \, \mathrm{d} r \, \mathrm{d} R\, \mathrm{d} I_* \, \mathrm{d}  v_*,
\end{multline}
with $\alpha>-1$, and where we have used the standard abbreviations 
\begin{equation}\label{f abb}
f':=f(t,x,v',I'), \ f'_*:=f(t,x,v'_*,I'_*), \ f_*:=f(t,x,v_*,I_*),
\end{equation}
 and quantities $v', I', v'_*, I'_*$ are described in the collision transformation $T$ from \eqref{coll mapping}. The  cross section  $\mathcal{B}^{nw}$  is supposed to satisfy the micro-reversibility conditions
\begin{equation}\label{Bnw}
\begin{split}
\mathcal{B}^{nw} :=\mathcal{B}^{nw}(v,v_*,I,I_*,R,r,\sigma) &= \mathcal{B}^{nw}(v',v'_*,I',I'_*,R',r',\sigma') \\ &=  \mathcal{B}^{nw}(v_*,v,I_*,I,R,1-r,-\sigma).
\end{split}
\end{equation}
Let us explain the terms involved in the strong form \eqref{Q non-weight}. First,  term $(1 - R) R^{\frac{1}{2}}$ is coming from the Jacobian of collision transformation computed in the Lemma \ref{Lemma coll mapping}. Then, renormalization of a distribution function $f$ by the factor $I^\alpha$ will allow to obtain the proper caloric equation of state, which causes  presence of  functions  $\phi_\alpha(r)$ and $ \psi_\alpha(R) $ aiming to ensure the invariance property of the measure, as shows the upcoming Lemma \ref{Lemma measure inv}. As we shall see later, $\alpha$ will be strongly connected to the molecule's number of degrees of freedom.
\begin{lemma}\label{Lemma measure inv}
	The measure
	\begin{align}\label{inv-measure}
	\mathrm{d}A=\mathcal{B}^{nw} \, \phi_\alpha(r)  \,  (1-R) R^{\frac{1}{2}} \psi_\alpha(R)\, I^{\alpha}  I_*^{\alpha} \, \mathrm{d} \sigma \, \mathrm{d} r \, \mathrm{d} R \,  \mathrm{d} I_* \, \mathrm{d}  v_*\, \mathrm{d} I \, \mathrm{d}  v
	\end{align}
	is invariant with respect to the changes
	\begin{align}
	(v,v_*,I,I_*,R,r,\sigma) & \leftrightarrow (v',v'_*,I',I'_*,R',r',\sigma'), \label{microreversibility changes prime}\\
	(v,v_*,I,I_*,R,r,\sigma)&\leftrightarrow (v_*,v,I_*,I,R,1-r,-\sigma). \label{microreversibility changes star}
	\end{align}
\end{lemma}
\begin{proof}
	The proof immediately follows from the property  \eqref{Bnw} of the cross section $\mathcal{B}^{nw}$, Lemma \ref{Lemma fun r R} and Jacobian of transformation \eqref{coll mapping} from Lemma \ref{Lemma coll mapping}.
\end{proof}
The strong form \eqref{Q non-weight} can be also written in the following manner,
\begin{multline}\label{Q non-weight pull out}
Q^{nw}(f,f)(v,I) = \int_{\mathbb{R}^3 \times [0,\infty)\times [0,1]^2 \times S^2} \left( \frac{f' f'_*}{\left( I' \, I'_* \right)^\alpha}  - \frac{f f_*}{\left(I \, I_*\right)^\alpha}\right) \\ \times \mathcal{B}^{nw} \, \phi_\alpha(r)  \,  (1-R) R^{\frac{1}{2}} \psi_\alpha(R)\, I^{\alpha}  I_*^{\alpha} \, \mathrm{d} \sigma \, \mathrm{d} r \, \mathrm{d} R\, \mathrm{d} I_* \, \mathrm{d}  v_*,
\end{multline}
obtained by pulling out the factor $(I\, I_*)^{\alpha}$.

\subsubsection{Collision operator $Q^{nw}$ in the weak form}  The choice of the functional space becomes evident in the definition of the weak form. Here we work in the plain $L^1$ space introduced in \eqref{space L^1 plain}.

\begin{lemma}[The weak form of the collision operator $Q^{nw}$] For  any test function $\chi(v,I)$ that makes the following left hand side meaningful, the collision operator \eqref{Q non-weight} takes the following weak form 
	\begin{multline}\label{weak form nw}
	\int_{\mathbb{R}^3 \times [0,\infty)} Q^{nw}(f, f)(v,I)  \, \chi(v,I) \, \mathrm{d}I \, \mathrm{d}v 
	\\
	= \frac{1}{2} \int_{\mathbb{R}^6 \times [0,\infty)^2 \times [0,1]^2 \times S^2 }   \frac{f f_*}{(I I_*)^\alpha}\left( \chi(v',I') + \chi(v'_*,I'_*) - \chi(v,I) - \chi(v_*,I_*)\right) \mathrm{d}A,
	\end{multline}
	with the measure $\mathrm{d}A$ from \eqref{inv-measure}.
\end{lemma}
\begin{proof} We integrate the collision operator \eqref{Q non-weight} against a test function  $\chi(v,I)$ with respect to $v$ and $I$ variables and then perform changes of variables \eqref{microreversibility changes prime} and \eqref{microreversibility changes star}. Using invariance properties of the measure  $\mathrm{d}A$  \eqref{inv-measure} stated  in Lemma \ref{Lemma measure inv}, we obtain
	\begin{align}
	& \int_{\mathbb{R}^3 \times [0,\infty)} 
	Q^{nw}(f,f)(v,I)  
	\chi(v,I)
	\mathrm{d}I	\, \mathrm{d} v \nonumber
	\\	&= 	\int_{\mathbb{R}^6 \times [0,\infty)^2 \times [0,1]^2 \times S^2 } \frac{ f f_* }{(I I_*)^\alpha} \left( \chi(v',I') -  \chi(v,I)\right)  \mathrm{d}A  \label{weak 1}
	\\	&= 	\int_{\mathbb{R}^6 \times [0,\infty)^2 \times [0,1]^2 \times S^2 } \frac{ f f_* }{(I I_*)^\alpha} \left( \chi(v'_*,I'_*) -  \chi(v_*,I_*)\right)  \mathrm{d}A, \nonumber
	\end{align}
	which yields the desired estimate \eqref{weak form}.
\end{proof}
The conservation laws at the microscopic level \eqref{micro CL} imply the annihilation of the weak form \eqref{weak form nw} for the conserved quantities. More precisely,
\begin{equation}\label{weak form coll inv}
\int_{\mathbb{R}^3 \times [0,\infty)} Q^{nw}(f, f)(v,I)  \left(  \begin{matrix} m \\  m v \\ \frac{m}{2} \left|v\right|^2 + I \end{matrix}  \right) \mathrm{d}I \, \mathrm{d}v 
\\
= 0.
\end{equation}
Any linear combination of  test functions $m$, $m v$ and $ \frac{m}{2} \left|v\right|^2 + I $ is called the collision invariant.\\

Our next goal is to formulate the H-theorem for the collision operator $Q^{nw}$. To that end, we  first define  the entropy production,
\begin{equation}\label{entr prod nw}
D^{nw}(f) = \int_{\mathbb{R}^3 \times [0,\infty) } Q^{nw}(f,f)(v,I) \log(f(v,I) I^{-\alpha}) \ \mathrm{d}I \mathrm{d}v,
\end{equation}
and then study its properties in the following theorem. 
\begin{theorem}[H-theorem]
	Let the cross section $\mathcal{B}^{nw}$ be positive almost everywhere, and let $f\geq 0$ be such that the collision operator $Q^{nw}(f,f)$ and the entropy production $D^{nw}(f)$ are well defined. Then the following properties hold,
	\begin{itemize}
		\item[i.] Entropy production is non-positive, that is
		\begin{equation}\label{H th non positive}
		D^{nw}(f) \leq 0.
		\end{equation}
		\item[ii.] The three following properties are equivalent
		\begin{itemize}
			\item[(1)]$D^{nw}(f) =0$,\\
			\item[(2)] $Q^{nw}(f,f) =0$  for all  $v \in \mathbb{R}^3, \ I\in [0,\infty)$,\\
			\item[(3)] There exists $n\geq 0$, $U \in \mathbb{R}^3$,  and  $T>0$, such that
			\begin{equation}\label{Maxwellian}
			f(v, I) = \frac{n}{Z(T)} \left( \frac{m}{2 \pi k T} \right)^{\frac{3}{2}} I^\alpha \ e^{- \frac{1}{k T} \left( \frac{m}{2} \left| v - U \right|^2 + I \right)} ,
			\end{equation}
			where $Z(T)$ is a partition (normalization) function 
			\begin{equation*}
			Z(T) = \int_{[0,\infty)} I^\alpha e^{-\frac{I}{k T} } \mathrm{d}I = (k T)^{\alpha +1} \Gamma(\alpha+1),
			\end{equation*}
			with	$\Gamma$ representing the  Gamma function.
		\end{itemize}
	\end{itemize}
\end{theorem}
The proof is given in \cite{LD-Bourgat}.

\subsection{Kinetic model in the  weighted setting}\label{Sec: w model}
The kinetic model in the weighted  setting originates from \cite{LD-Toulouse, DesMonSalv}. In this case the  distribution function 
\begin{equation*}
g := g(t, x, v, I) \geq 0,
\end{equation*}
satisfies the Boltzmann equation 
\begin{equation}\label{BE weight}
\partial_t g + v \cdot \nabla_{x} g =Q^w(g,g)(v,I),
\end{equation}
where $Q^w(g,g)$ is the collision operator that acts only on $(v,I)$ variables and is described below.
\subsubsection{Collision operator $Q^w$ in the strong form} The weighted setting is related to the weighted $L^1$ space \eqref{space L^1 weight}, and in this case the collision operator is defined as
\begin{multline}\label{Q weight}
Q^w(g,g)(v,I) = \int_{\mathbb{R}^3 \times [0,\infty) \times [0,1]^2 \times S^2} \left(g' g'_* - g g_*\right) \\ \times \mathcal{B}^w (1-R) R^{\frac{1}{2}} \frac{1}{\varphi(I)} \mathrm{d} \sigma \, \mathrm{d} r \, \mathrm{d} R\, \mathrm{d} I_* \, \mathrm{d}  v_*,
\end{multline}
where we have used the standard conventions as in \eqref{f abb},  with  the primed quantities  from \eqref{coll mapping},  and the cross section
\begin{equation}\label{Bw}
\begin{split}
\mathcal{B}^w :=\mathcal{B}^w(v,v_*,I,I_*,R,r,\sigma) &= \mathcal{B}^w(v',v'_*,I',I'_*,R',r',\sigma') \\ &=  \mathcal{B}^w(v_*,v,I_*,I,R,1-r,-\sigma).
\end{split}
\end{equation}
The factor $(1-R) R^{\frac{1}{2}}$ is Jacobian of the collision transformation \eqref{coll mapping-3}. The measure $\varphi(I)$ aims at capturing the features of polyatomic gases at the macroscopic level, and notably to provide an agreement with the caloric equation of state. Contrary to the non-weighted setting, this measure is not introduced a priori, which theoretically gives  a room to obtain a general equation for polytropic or non-polytropic gases, corresponding to linear or non-linear dependence of the macroscopic internal energy upon temperature, respectively. 

\subsubsection{Collision operator $Q^w$ in the weak form}   For the weighted setting, the weak form of collision operator is obtained by means of the integration against the weight function $\varphi(I)$, as  described in the upcoming Lemma \ref{Lemma weak w}.

\begin{lemma}\label{Lemma weak w} For  any test function $\chi(v,I)$ that makes the following left hand side meaningful, the collision operator \eqref{Q weight} has the following weak form 
	\begin{multline}\label{weak form}
	\int_{\mathbb{R}^3 \times [0,\infty) } Q^w(g, g)(v,I)  \, \chi(v,I) \, \varphi(I) \, \mathrm{d}I \, \mathrm{d}v 
	\\
	= \frac{1}{2} \int_{\mathbb{R}^6 \times [0,\infty)^2 \times [0,1]^2 \times S^2 }   g g_* \left( \chi(v',I') + \chi(v'_*,I'_*) - \chi(v,I) - \chi(v_*,I_*)\right)
	\\
	\times \mathcal{B}^w \, (1-R) R^{\frac{1}{2}}\, \mathrm{d} \sigma \, \mathrm{d} r \, \mathrm{d} R \,  \mathrm{d} I_* \, \mathrm{d}  v_*\, \mathrm{d} I \, \mathrm{d}  v,
	\end{multline}
	where primed quantities as functions of non-primed ones are given  in \eqref{coll mapping}, and $\mathcal{B}^w$ is from \eqref{Bw}.
\end{lemma}
\begin{proof}
	We first integrate the strong form \eqref{Q weight} against a suitable test function $\chi(v,I)$ in the velocity-internal energy space $\mathbb{R}^3 \times [0,\infty)$ with the weight $\varphi(I)$ in $I$. Then we change the variables, first we interchange primes and non-primes \eqref{microreversibility changes prime} and then we replace particles by means of  \eqref{microreversibility changes star}. This  gives
	\begin{align*}
	& \int_{\mathbb{R}^3 \times [0,\infty)} 
	Q^w(g,g)(v,I)  
	\chi(v,I) \varphi(I)
	\mathrm{d}I	\, \mathrm{d} v \nonumber
	\\	&= 	\int_{\mathbb{R}^6 \times [0,\infty)^2 \times [0,1]^2 \times S^2 } g g_* \left( \chi(v',I') -  \chi(v,I)\right) \nonumber\\  &\hspace*{5cm}\times \mathcal{B}^w \, (1-R) R^{\frac{1}{2}}\, \mathrm{d} \sigma \, \mathrm{d} r \, \mathrm{d} R \,  \mathrm{d} I_* \, \mathrm{d}  v_*\, \mathrm{d} I \, \mathrm{d}  v, \label{weak 0}
	\\	&= 	\int_{\mathbb{R}^6 \times [0,\infty)^2 \times [0,1]^2 \times S^2 } g g_* \left( \chi(v'_*,I'_*) -  \chi(v_*,I_*)\right) \nonumber \\  &\hspace*{5cm}\times  \mathcal{B}^w \, (1-R) R^{\frac{1}{2}}\, \mathrm{d} \sigma \, \mathrm{d} r \, \mathrm{d} R \,  \mathrm{d} I_* \, \mathrm{d}  v_*\, \mathrm{d} I \, \mathrm{d}  v,  \nonumber
	\end{align*}
where we have used invariance of the cross section $\mathcal{B}^w$ stated in \eqref{Bw} and Jacobian of the transformation \eqref{coll mapping-3}, which concludes the relation \eqref{weak form}.
\end{proof}
Microscopic conservation laws  \eqref{micro CL} imply that the weak form \eqref{weak form} vanishes when test functions are chosen as $m$, $mv$ and $\frac{m}{2} \left|v\right|^2 + I $,
\begin{equation*}\label{weak form coll inv w}
\int_{\mathbb{R}^3 \times [0,\infty)} Q^w(g, g)(v,I)  \left(  \begin{matrix} m \\  m v \\ \frac{m}{2} \left|v\right|^2 + I \end{matrix}  \right) \mathrm{d}I \, \mathrm{d}v 
\\
= 0.
\end{equation*}
We refer to these test functions and  to their any linear combination  as  collision invariants.\\

We now  formulate the H-theorem for the collision operator $Q^w$. We first define the entropy production,
\begin{equation*}
D^{w}(g) = \int_{\mathbb{R}^3 \times [0,\infty)} Q^w(g,g)(v,I)\, \log(g(v,I)) \ \varphi(I) \, \mathrm{d}I \mathrm{d}v,
\end{equation*}
\begin{theorem}[H-theorem]
	Let the  cross section $\mathcal{B}^{w}$ be positive almost everywhere, and let $g\geq 0$ such that the collision operator $Q^{w}(g,g)$ and entropy production $D^{w}(g)$ are well defined. Then the following properties hold
	\begin{itemize}
		\item[i.] Entropy production is non-positive, that is
		\begin{equation*}
		D^{w}(g) \leq 0.
		\end{equation*}
		\item[ii.] The three following properties are equivalent
		\begin{itemize}
			\item[(1)]$D^{w}(g) =0$,\\
			\item[(2)] $Q^{w}(g,g) =0$  for all  $v \in \mathbb{R}^3, \ I \in [0, \infty)$,\\
			\item[(3)] There exists $n\geq 0$, $U \in \mathbb{R}^3$,  and  $T>0$, such that
			\begin{equation*}
			g(v, I) = \frac{n}{Z(T)} \left( \frac{m}{2 \pi k T} \right)^{\frac{3}{2}} \ e^{- \frac{1}{k T} \left( \frac{m}{2} \left| v - U \right|^2 + I \right)} ,
			\end{equation*}
			where $Z(T)$ is a  normalization function 
			\begin{equation*}
			Z(T) = \int_{[0,\infty)}  e^{-\frac{I}{k T} }  \varphi(I) \, \mathrm{d}I.
			\end{equation*}
		\end{itemize}
	\end{itemize}
\end{theorem}
The proof can be found in \cite{DesMonSalv}.

\section{Comparison of the continuous kinetic model in different settings}\label{Sec: Comp}

In order to compare the continuous kinetic model in two different settings, we start with the weighted one presented in Section \ref{Sec: w model}.  We firstly change the distribution function $g$ so that the weight $\varphi(I)$ is  detached from it. More precisely, we introduce the distribution function $f$ by means of
\begin{equation}\label{distr fun change}
f(t,x,v,I) = g(t,x,v,I) \varphi(I).
\end{equation}
Then the Boltzmann equation \eqref{BE weight} written in terms of $g$ can be  rewritten  in terms of  $f$,
\begin{multline}\label{BE Q non weight general}
\partial_t f + v \cdot \nabla_{x} f =  \int_{\mathbb{R}^3 \times [0,\infty)\times [0,1]^2 \times S^2} \left( f' f'_* \ \frac{\varphi(I) \, \varphi(I_*)}{\varphi(I') \, \varphi(I'_*)}- f f_*\right) \\
\times\frac{\mathcal{B}^w}{\varphi(I) \, \varphi(I_*)}  (1-R) R^{\frac{1}{2}} \, \mathrm{d} \sigma \, \mathrm{d} r \, \mathrm{d} R\, \mathrm{d} I_* \, \mathrm{d}  v_*,
\end{multline}
which	coincides with the Boltzmann equation from \cite{LD-Bisi}, Remark 1. However, the effective collision cross section 
\begin{equation}\label{effective cross section}
\frac{\mathcal{B}^w}{\varphi(I) \, \varphi(I_*)},
\end{equation}
as it was named in \cite{LD-Bisi}, in general does not satisfy the micro-reversibility assumptions.  

In order to overcome this drawback and to obtain the model in the non-weighted setting described in Section \ref{Sec: n-w model}, we need to make a choice of the function $\varphi(I)$, and so we take
\begin{equation}\label{fi}
\varphi(I) = I^\alpha.
\end{equation}
With this choice of the weight function $\varphi(I)$, the Boltzmann equation \eqref{BE Q non weight general} becomes
\begin{multline*}
\partial_t f + v \cdot \nabla_{x} f =  \int_{\mathbb{R}^3 \times [0,\infty)\times [0,1]^2 \times S^2} \left( f' f'_*  \left(\frac{I \, I_*}{I' \, I'_*} \right)^{\alpha}- f f_*\right) \\ \times \frac{\mathcal{B}^w}{I^\alpha I_*^\alpha }  (1-R) R^{\frac{1}{2}} \, \mathrm{d} \sigma \, \mathrm{d} r \, \mathrm{d} R\, \mathrm{d} I_* \, \mathrm{d}  v_*.
\end{multline*}
Now, it is clear that in order to recover formula in the non-weighted  setting \eqref{BE non-weight}-\eqref{Q non-weight}, we need to multiply and divide by a factor $\left( r(1-r) \right)^\alpha (1-R)^{2\alpha} =  \phi_\alpha(r) \, \psi_\alpha(R)$, 
\begin{multline*}
\partial_t f + v \cdot \nabla_{x} f =  \int_{\mathbb{R}^3 \times [0,\infty)\times [0,1]^2 \times S^2} \left( f' f'_*  \left(\frac{I \, I_*}{I' \, I'_*} \right)^{\alpha}- f f_*\right)
\\ \times \frac{\mathcal{B}^w}{I^\alpha I_*^\alpha \phi_\alpha(r) \, \psi_\alpha(R)} \phi_\alpha(r) \, \psi_\alpha(R) (1-R) R^{\frac{1}{2}} \, \mathrm{d} \sigma \, \mathrm{d} r \, \mathrm{d} R\, \mathrm{d} I_* \, \mathrm{d}  v_*.
\end{multline*}
This factor makes the effective cross section \eqref{effective cross section} for the choice \eqref{fi} micro-reversible, and so we are led to define the new cross section
\begin{equation}\label{Bnw preko Bw}
\mathcal{B}^{nw} = \frac{\mathcal{B}^w}{I^\alpha I_*^\alpha \, \phi_\alpha(r) \, \psi_\alpha(R)},
\end{equation}
that leads to the non-weighted setting as in  \eqref{BE non-weight}-\eqref{Q non-weight}. \\

It is worthwhile to remark that in order to pass from the weighted to the non-weighted setting or vice versa,  it is not enough to only change distribution function by detaching or attaching   the weight. In addition, we need to reformulate the cross section using \eqref{Bnw preko Bw}. Therefore, one needs to be careful in choosing the cross section, since the collision operator \eqref{Q non-weight} hides the micro-reversible part \eqref{fun r R} in its cross section $\mathcal{B}^{nw}$.   In other words, cross section $\mathcal{B}^{nw}$ that appears in the kinetic model in the non-weighted setting is not the same as $\mathcal{B}^w$ from the kinetic model in the weighted setting. \\

Reformulation of the cross section stated in \eqref{Bnw preko Bw} has some mathematical consequences as well.  Namely, for the choice $\varphi(I) = I^\alpha$ the strong form of collision operator in the weighted setting \eqref{Q weight} has  singularity at zero for the variable $I$, while the collision operator in the non-weighted setting \eqref{Q non-weight} does  not have this drawback. 
The formula \eqref{Bnw preko Bw} reveals the reason for such a behavior:   the cross section in \eqref{Q non-weight} is multiplied by a factor that removes the singularity.\\

Therefore, we conclude that the equivalence of the models \eqref{BE weight}-\eqref{Q weight}  and  \eqref{BE non-weight}-\eqref{Q non-weight}  holds after the distribution function is renormalized as in \eqref{distr fun change}  and the cross section is reformulated by using the formula \eqref{Bnw preko Bw} for the choice \eqref{fi} of the weight function.

\section{Macroscopic models for a rarefied polyatomic gas}\label{Sec: macro}

As it is very well known, the kinetic models can provide  the macroscopic ones, using the so called moment method. The basic idea is to build infinite hierarchies of moment equations by integrating the Boltzmann equation over the microscopic molecular variables space -- velocity space in the case of a monatomic gas leading to one hierarchy of moments, or velocity and microscopic  internal energy space for a polyatomic gas causing the two hierarchies of moments \cite{MPC-Rugg-Sim}. These infinite  hierarchies are cut at some order of moments, that yield a non-closed system of moment equations. One of the possible ways to close the system  is to formulate a variational problem, the \emph{maximum entropy principle}, that seeks for a distribution function which maximizes the physical entropy subject to some constraints.  These constraints are actually   macroscopic moment densities related to the physical process at hand. Determination of the distribution function allows to obtain non-convective fluxes as functions of moment densities. Then it remains to calculate the production terms, when the choice of the cross section becomes a crucial aspect.\\

In Sections \ref{Sec: macro six} and \ref{Sec: macro 14} we study six and fourteen fields models, respectively. The six moments model corresponds to the usual equilibrium macroscopic observables, mass $\rho$, momentum $\rho \, U$ and energy density $\frac{1}{2}\rho \left|U\right|^2 + \rho e$ , and the dynamic pressure $\Pi$ as a dominant non-equilibrium variable. The fourteen moments model extends the  list of non-equilibrium effects, by taking into account the stress tensor $-p_{ij}$, $i,j\in \left\{ 1, 2, 3 \right\}$, and  the heat flux $q_i$, $i=1,2,3$. Our goal is to use the maximum entropy principle and to establish both six and fourteen moments models starting from the Boltzmann  equation in the non-weighted setting described in the Section  \ref{Sec: n-w model}.  \\

One of the main tools  in the maximum entropy principle procedure  is the physical entropy density defined as follows.
\begin{definition}[Physical entropy] For a distribution function $f\geq 0$ we define the
	physical entropy density
	\begin{equation}\label{entripy phy}
		h = - k   \int_{\mathbb{R}^3 \times [0,\infty)} f \, \log(f I^{-\alpha}) \, \mathrm{d}I \, \mathrm{d}v,
	\end{equation}
	where $k$ is the Boltzmann constant and $\alpha>-1$.
\end{definition}
The entropy law is obtained by integration of the Boltzmann equation \eqref{BE non-weight} against the test function $\log\left( f\, I^{-\alpha} \right)$ multiplied by  the factor $-k$, i.e.
\begin{equation*}
	\partial_t h + \sum_{j=1}^{3} \partial_{x_j} h_j = \Sigma, 
	\end{equation*}
where $h_j$ is the flux of  entropy density in the direction $x_j$ and $\Sigma$   is the  entropy production density, defined from the kinetic theory point of view via
\begin{equation}\label{entropy flux prod}
h_j = - k   \int_{\mathbb{R}^3 \times [0,\infty)} v_j f \, \log(f I^{-\alpha}) \, \mathrm{d}I \, \mathrm{d}v, \quad \Sigma = - k D^{nw}(f),
\end{equation}
where $D^{nw}(f)$  was already introduced in \eqref{entr prod nw}.\\
	
\subsection{Local equilibrium state}
Both six and fourteen fields models are designated for non-equilibrium processes. We introduce the notion of local equilibrium using the idea of  maximum entropy principle. Namely, we seek for a distribution function such that the physical entropy \eqref{entripy phy} is maximized subject to the prescribed macroscopic densities obtained as its moments against the collision invariants \eqref{weak form coll inv}. Formally, the functional space in which the distribution function is looked for  is the space  of  integrable functions weighted with the polynomials corresponding to the  collision invariants   \eqref{weak form coll inv} which ensures well-defined macroscopic densities.  

In fact, macroscopic densities are defined using the peculiar velocity $c$, which is a relative velocity  of the particle of velocity $v$ with respect to the macroscopic gas velocity $U$, namely 
\begin{equation*}
	c=v-U.
\end{equation*}

\begin{lemma}[Local equilibrium distribution function]
	The distribution function that  solves the following problem
	\begin{alignat}{3}\label{MEP eq: var problem}
	\max_f& \quad h = - k   \int_{\mathbb{R}^3 \times [0,\infty)} f \, \log(f I^{-\alpha}) \, \mathrm{d}I \, \mathrm{d}v  \nonumber\\
	\text{s.t.}  & \left( \begin{array}{c}
	\rho \\
	0_i \\
	\left( \alpha + \frac{5}{2} \right) p
	\end{array}\right)   =  \int_{\mathbb{R}^3 \times [0,\infty)}
	\left(
	\begin{array}{c}
	m \\
	m\, c_{i} \\
	\frac{m}{2} \left| {c} \right|^2 + I
	\end{array}
	\right)\,
	f \, \mathrm{d} I \, \mathrm{d}{v},
	\end{alignat}
	with the classical thermodynamic (hydrostatic) pressure $p$  related to the temperature $T$ through
	\begin{equation}\label{hydro p}
	p=	\frac{\rho}{m}\, k\, T,
	\end{equation}
is  the local  equilibrium distribution function,
	\begin{equation}\label{Maks}
	f_M(t, x, v, I) := I^\alpha  \frac{\rho}{m} \left( \frac{m}{2 \pi k T} \right)^{\frac{3}{2}} \frac{1}{\Gamma(\alpha + 1)} \frac{1}{\left(k T\right)^{\alpha+1}}
 \, e^{-\frac{1}{kT} \left( \frac{m}{2}  \left|v-U\right|^2 +  I \right)},
	\end{equation}
with $\alpha>-1$.
\end{lemma}
\begin{proof}
	We follow the classical procedure of maximum entropy principle. Namely, we first 
introduce Lagrange multipliers $\lambda^{(0)}$, $\lambda^{(1)}_i$ and $\mu^{(0)}$ that  correspond to the constraints \eqref{MEP eq: var problem}.  Then the extended functional reads
\begin{multline*}
\mathcal{L} = \int_{\mathbb{R}^3 \times [0,\infty)  } \Bigg\{  - k  f \, \log(f I^{-\alpha})  
\\  - f \left( \lambda^{(0)} m + \sum_{i=1}^3  \lambda^{(1)}_i  m\, c_i + \mu^{(0)} \left(	\frac{m}{2} \left| {c} \right|^2 + I\right) \right)\Bigg\} \mathrm{d} I \, \mathrm{d}{c}.
\end{multline*}
The solution of the Euler-Lagrange equation $\delta\mathcal{L}/\delta f$ is given by
\begin{equation*}
f=I^\alpha e^{-1-\frac{m}{k} \lambda^{(0)}  - \frac{m}{k} \sum_{i=1}^3  \lambda^{(1)}_i  c_i - \frac{1}{k}  \mu^{(0)}  \left(  \frac{m}{2} \left| {c} \right|^2 + I \right)  }.
\end{equation*}
Plugging this form into the constraints of the problem \eqref{MEP eq: var problem} we get a system of algebraic equations whose solution allows to express  Lagrange multipliers in terms of macroscopic densities, which implies the solution \eqref{Maks}.
\end{proof}
\begin{remark}\label{Rem}
We mention that Lagrange multipliers do not depend on the choice of the functional space, and thus they coincide with the ones obtained using the kinetic  model in the weighted setting.
\end{remark}

The integration of the Boltzmann equation  against the collision invariants \eqref{weak form coll inv}  over the velocity-microscopic internal energy space yields the Euler system of equations for macroscopic densities appearing as the constraints of the problem \eqref{MEP eq: var problem}. This system is conservative  (i.e. production terms are all zero) which can be concluded from the kinetic theory point of view by  the vanishing  properties of the collision operator weak form \eqref{weak form coll inv}. 

Beyond test functions corresponding to the collision invariants, collision operator will have non-zero weak form and in order to compute production terms the cross-section needs to be chosen. 

\subsection{Cross-section model} 
In this paper,  the production terms are calculated for the following choice of the cross section 
\begin{multline}\label{model 3}
	\mathcal{B}^{nw}(v, v_*, I, I_*, r, R, \sigma)  \\
	= b\!\,\big(\tfrac{v-v_*}{\left| v-v_* \right|} \cdot \sigma\big)  \left( R^{\frac{\gamma}{2}} |v-v_*|^\gamma  + \left( r (1-R)\frac{I}{m} \right)^{\frac{\gamma}{2}} + \left( (1-r) (1-R)\frac{I_*}{m} \right)^{\frac{\gamma}{2}} \right),  
\end{multline}
with $ \gamma>0$,  and any $v, v_* \in \mathbb{R}^3$, $I,I_* \in [0,\infty)$, $r, R \in [0,1]$ and $\sigma \in S^2$. The function $b$ will be assumed integrable over the unit sphere $S^2$ in the case of six fields equations, whereas will be taken constant for the fourteen moments model. \\

The motivation for this choice of the cross-section lies in the recent mathematically rigorous  result from  \cite{MPC-IG-poly}. Namely, this model of the cross-section corresponds to the   cross section model 3 introduced  in \cite{MPC-IG-poly}, for which the existence and uniqueness of the solution to the Cauchy problem for the  Boltzmann equation \eqref{BE non-weight} in the space homogeneous setting with suitable initial data is proven. \\

Besides this mathematical motivation, a physical intuition can be provided by studying the collision frequency defined in \cite{MPC-IG-poly} for the polyatomic gas model. 

\subsubsection{Collision frequency} The collision frequency corresponding to the cross-section $	\mathcal{B}^{nw}$ from \eqref{model 3} for any distribution function $f$ is defined  with 
\begin{equation}\label{coll freq}
	\nu_{\gamma,\alpha}(v,I) = \int_{\mathbb{R}^3 \times [0,\infty) \times [0,1]^2 \times S^2} f_* \,	\mathcal{B}^{nw} \phi_\alpha(r) \, \psi_\alpha(R) (1-R) R^{\frac{1}{2}} \, \mathrm{d} \sigma \, \mathrm{d} r \, \mathrm{d} R\, \mathrm{d} I_* \, \mathrm{d}  v_*.
\end{equation}
The collision frequency can be explicitly computed in a state near equilibrium when   the distribution function $f$ is replaced by the  local Maxwellian $f_M$ introduced in  \eqref{Maks}, yielding to the physical interpretation of the cross-section $\mathcal{B}^{nw}$.  Indeed, choosing \eqref{model 3}, after the calculation presented in Appendix \ref{App coll fre}, we obtain the following expression 
\begin{multline} \label{coll freq1}
	\nu_{\gamma,\alpha}(v,I) =  \frac{\rho}{m}  \left\| b \right\|_{L^1(\mathrm{d}\sigma)}  \left( \frac{p}{\rho}  \right)^{\frac{\gamma}{2}}  \frac{\Gamma (\alpha +1)}{\Gamma\! \left(\frac{4 \alpha +\gamma +7}{2}\right)}      \\ \times
\left( \Gamma (\alpha +1) \Gamma \!\left(\frac{\gamma +3}{2}\right)^2 \frac{2^{\frac{\gamma}{2}+1}}{\sqrt{\pi}}     \, e^{-\frac{m}{2 k T} \vert v - U\vert^2   } \ _1F_1\! \left(\frac{\gamma+3}{2}, \frac{3}{2}, \frac{m \vert v- U \vert ^2}{2 k T}\right) 
\right. \\
\left. +  \frac{\sqrt{\pi } }{2} \Gamma \!\left(\alpha +\frac{\gamma }{2}+1\right) \left(  \left(\frac{I}{k T}\right)^\frac{\gamma}{2}+  \frac{\Gamma \!\left(\alpha+\frac{\gamma}{2}+1\right)}{\Gamma(\alpha+1) } \right) \right),   
\end{multline}
with the hypergeometric function $ _1F_1$ defined by  \eqref{hyp geo} in the Appendix.

\begin{figure}[t]
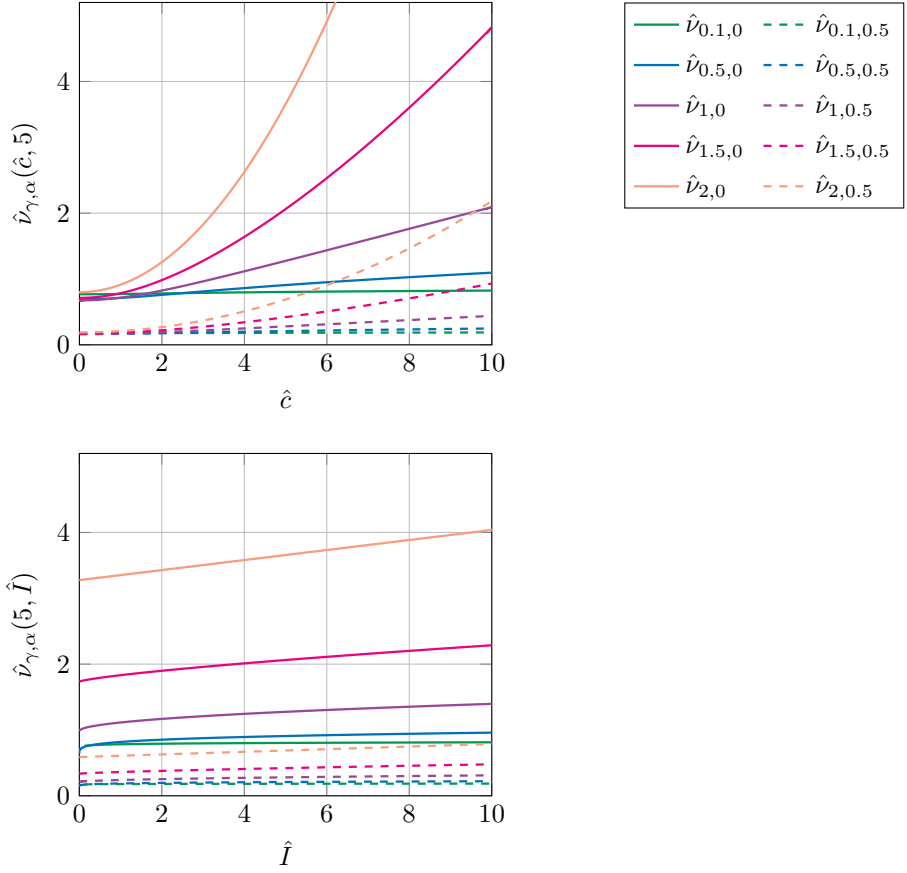

	\begin{center}


\end{center}
	 		\caption{The collision frequency $\hat{\nu}_{\gamma,\alpha}(\hat{c},\hat{I})$ defined in \eqref{coll freq2} as a function of  the dimensionless 
	 			 peculiar speed $\hat{c}  $ and   internal energy $\hat{I} $ introduced in  \eqref{dimless}  for certain values of $\gamma$ ranging in the interval $(0,2]$, and for   $\alpha=0$ corresponding to linear molecules (solid line)  and $\alpha=\frac{1}{2}$ corresponding  to non-linear molecules with translational and rotational degrees of freedom (dashed line).  }
	\label{figure-prvaslika}	
	\end{figure}

Introducing dimensionless peculiar speed and microscopic internal energy 
\begin{equation}\label{dimless}
	\hat{c} = \sqrt{\frac{m}{k T}} \left|v - U\right|, \qquad \hat{I}=\frac{1}{k T} I, 
\end{equation}
we can define the dimensionless collision frequency,
\begin{multline} \label{coll freq2}
	\hat{\nu}_{\gamma,\alpha}(\hat{c},\hat{I}) =   \frac{\Gamma (\alpha +1)}{\Gamma\! \left(\frac{4 \alpha +\gamma +7}{2}\right)}     
	\left( \Gamma (\alpha +1) \Gamma \!\left(\frac{\gamma +3}{2}\right)^2 \frac{2^{\frac{\gamma}{2}+1}}{\sqrt{\pi}}     \, e^{-\frac{\hat{c}^2}{2 }   } \ _1F_1\! \left(\frac{\gamma+3}{2}, \frac{3}{2}, \frac{\hat{c}^2}{2 }\right) 
	\right. \\
	\left. +  \frac{\sqrt{\pi } }{2} \Gamma \!\left(\alpha +\frac{\gamma }{2}+1\right) \left(  \hat{I}^\frac{\gamma}{2}+  \frac{\Gamma \!\left(\alpha+\frac{\gamma}{2}+1\right)}{\Gamma(\alpha+1) } \right) \right).
\end{multline}
Figure \ref{figure-prvaslika} represents the collision frequency $	\hat{\nu}_{\gamma,\alpha}(\hat{c},\hat{I})$  for certain values of parameters $\alpha$ and $\gamma$. One can observe that, for the fixed value of $\alpha$, the collision frequency is a monotonically increasing function of both the velocity and internal energy and the growth is stronger for larger $\gamma>0$. This means that fast particles and/or high energy particles are more likely to collide. Such a  behavior is known for a monatomic gas as well \cite{Str}.  

\section{Six fields model}\label{Sec: macro six}

The six fields model corresponds to the moment equations obtained by integrating the Boltzmann equation \eqref{BE non-weight} with respect to the microscopic set of variables -- molecular velocity $v$ and microscopic internal energy $I$, against six test functions,
\begin{equation}\label{test 6}
m, \ mv, \ m \left|v\right|^2, \ \frac{m}{2} \left| v\right|^2 + I.
\end{equation}
Introducing the peculiar velocity $c=v-U$,  we define densities of macroscopic observables that correspond to these test functions as follows
\begin{equation}\label{6 densities}
\left( \begin{array}{c}
	\rho \\
	\rho \, U \\
	3(p + \Pi) \\
\rho \, e
\end{array}\right)   =  \int_{\mathbb{R}^3 \times \mathbb{R}_+}
\left(
\begin{array}{c}
	m \\
	m\, v \\
	m\, \left| {c} \right|^2\\
	\frac{m}{2} \left| {c} \right|^2 + I
\end{array}
\right)\,
f \, \mathrm{d} I \, \mathrm{d}{v}.
\end{equation}
Now integration of the Boltzmann equation against test functions \eqref{test 6} leads to the following set of  six moments equations
\begin{equation}\label{6 fields non-closed}
\begin{alignedat}{2}
\partial_t \rho  +  \sum_{j=1}^{3} \partial_{x_j}(\rho\, U_j)&=0,\\
\partial_t \rho \, U_i + \sum_{j=1}^{3} \partial_{x_j} \left(\rho U_iU_j +p_{ij}\right)&=0,\\
\partial_t \left( \rho \vert U\vert^2 + 3(p+ \Pi) \right)  \hphantom{\sum_{j=1}^{3} \partial_{x_j} \left\{ U_j  \left(\rho \vert U\vert^2   +   3(p+ \Pi)\right)  + 2 \sum_{i=1}^3p_{ji} U_i  + \sum_{i=1}^3 p_{iij}        \right\}}&\\
+  \sum_{j=1}^{3} \partial_{x_j} \left\{ U_j  \left(\rho \vert U\vert^2   +   3(p+ \Pi)\right)  + 2 \sum_{i=1}^3p_{ji} U_i  + \sum_{i=1}^3 p_{iij}        \right\}&=\mathcal{P},\\
\partial_t \left( \frac{1}{2}\rho \vert U\vert^2 + \rho e  \right) +  \sum_{j=1}^{3} \partial_{x_j} \left\{  U_j\left(\frac{1}{2}\rho\vert\mathrm{U}\vert^2 + \rho\,e \right) + \sum_{i=1}^3p_{ji} U_i +q_j                   \right\}&=0,
\end{alignedat}
\end{equation}
for  $i=1,2,3$,  and where we have assumed the following relations
\begin{equation}\label{polytrop}
\sum_{i=1}^3 p_{ii} = 3 (p+\Pi), \quad \rho e = \left(\alpha+ \frac{5}{2}\right) p,
\end{equation}
and \eqref{hydro p}. 
The undetermined non-convective fluxes are defined as
\begin{equation}\label{6 fields non-conv flux}
\left( \begin{array}{c}
p_{ij} \\
\sum_{i=1}^3 p_{iij}\\
q_j
\end{array}\right)   =  \int_{\mathbb{R}^3 \times \mathbb{R}_+}
\left(
\begin{array}{c}
m \, c_i\, c_j \\
m\, \left| {c} \right|^2 c_j \\
\left( \frac{m}{2} \left| {c} \right|^2 + I \right) c_j
\end{array}
\right)\,
f \, \mathrm{d} I \, \mathrm{d}{c},
\end{equation}
while the production term reads
	\begin{equation}
\label{productionPii:equ}
\mathcal{P}=\int_{\mathbb{R}^3 \times [0,\infty) } m\, \vert v\vert^2 Q^{nw}(f,f)(v,I) \, \mathrm{d}I \mathrm{d}v.
\end{equation} 
The goal of this section is to provide a closure to the system \eqref{6 fields non-closed} via the maximum entropy principle, which is achieved by determining the six fields distribution function in the Lemma \ref{Prop: f6}, and then calculating non-convective fluxes \eqref{6 fields non-conv flux}, as much as the production term \eqref{productionPii:equ} for a specific choice of the cross section \eqref{model 3}. The obtained results are compared with the theory of extended thermodynamics  in Section \ref{Sec ET}.

\begin{lemma}[Six moments distribution function]\label{Prop: f6}
Solution of  the maximum entropy principle
\begin{alignat*}{3}
\max_f& \quad h = - k   \int_{\mathbb{R}^3 \times [0,\infty)} f \, \log(f I^{-\alpha}) \, \mathrm{d}I \, \mathrm{d}v \nonumber\\
\text{s.t.}  & \left( \begin{array}{c}
\rho \\
0_i \\
3(p + \Pi) \\
\left( \alpha+ \frac{5}{2}\right)p
\end{array}\right)   =  \int_{\mathbb{R}^3 \times \mathbb{R}_+}
\left(
\begin{array}{c}
m \\
m\, c_{i} \\
m\, \left| {c} \right|^2\\
\frac{m}{2} \left| {c} \right|^2 + I
\end{array}
\right)\,
f \, \mathrm{d} I \, \mathrm{d}{c}.
\end{alignat*}
is given with
\begin{multline}
\label{f six}
\hat{f}^6 = I^\alpha \frac{\rho}{m} \left( \frac{m}{2 \pi k T} \right)^{\frac{3}{2}} \frac{1}{\left( 1 + \frac{\Pi}{p} \right)^{\frac{3}{2}}} \frac{1}{\Gamma(\alpha + 1)} \frac{1}{\left(k T\right)^{\alpha+1}} \left(\frac{1}{1-\frac{3}{2} \frac{\Pi}{(\alpha+1)p}}\right)^{\alpha+1} 
\\
\times  e^{-\frac{1}{kT} \left( \frac{m}{2} \frac{1}{\left(1+\frac{\Pi}{p}\right) } \left|c\right|^2 + \left( \frac{1}{1-\frac{3}{2} \frac{\Pi}{(\alpha+1)p}} \right) I \right)},
\end{multline}
with $\alpha>-1$, that provides convergent moments if
\begin{equation}\label{range Pi/p}
-1 < \frac{\Pi}{p} < \frac{2}{3} \left(\alpha+1\right).
\end{equation}
\end{lemma}
The proof is very similar to the one in \cite{Rugger-MEP, MPC-Madj-Sim}, by virtue of Remark \ref{Rem}. \\

This distribution function allows to close the system of equations \eqref{6 fields non-closed} for a particular choice of the cross section, as states the following proposition.
\begin{proposition}[Six fields system of equations ] Closed system of equations for six fields reads
	\begin{equation*}\label{oiler:sitem6}
\begin{alignedat}{2}
	\partial_t \rho  +  \sum_{j=1}^{3} \partial_{x_j}(\rho\, U_j)&=0,\\
	\partial_t \rho U_i + \sum_{j=1}^{3} \partial_{x_j} \left(\rho U_iU_j + (\Pi+ p)\,\delta_{ij}\right)&=0,\\
	\partial_t \left( \rho \vert U\vert^2 + 3(p+ \Pi) \right) +  \sum_{j=1}^{3} \partial_{x_j} \left\{ U_j  \left(\rho \vert U\vert^2   +   5(p+ \Pi)\right)            \right\}&=\mathcal{P},\\
	\partial_t \left( \frac{1}{2}\rho \vert U\vert^2 +\left( \alpha+ \frac{5}{2}\right)p \right) +  \sum_{j=1}^{3} \partial_{x_j} \left\{  U_j\left(\rho\frac{\vert\mathrm{U}\vert^2}{2}  + \left( \alpha+ \frac{7}{2}\right)p +\Pi  \right)                   \right\}&=0,
	\end{alignedat}
	\end{equation*}
for  $i=1,2,3$, 	where the production term $\mathcal{P}$ for the choice of the cross section \eqref{model 3} with the function $b \in L^1(\mathrm{d}\sigma)$, is given by
	\begin{equation}\label{P iz Prop}
	\mathcal{P} = - \mathcal{C}_\mathcal{P}    \frac{\rho^{2 }}{m} \left( \frac{p}{\rho}\right)^{\frac{\gamma}{2}+1}  \left\| b \right\|_{L^1(\mathrm{d}\sigma)}    \frac{\Pi}{p},
	\end{equation}
where $\mathcal{C}_\mathcal{P} $ is the positive function of $\frac{\Pi}{p}$ on the domain \eqref{range Pi/p},
	\begin{multline}\label{Cp}
\mathcal{C}_\mathcal{P} = \sqrt{\frac{2}{\pi}}  \frac{\left(\alpha+\frac52\right)}{\left(\alpha+1\right)}  \frac{1}{\Gamma\!\left( \frac{4\alpha+\gamma+9}{2} \right)} \\ \times \left( k_1 \left(2\left(1+\frac{\Pi}{p}\right)\right)^\frac{\gamma}{2} + k_2 \left(1-\frac{3}{2(\alpha+1)} \frac{\Pi}{p} \right)^{\frac{\gamma}{2}} \right),
\end{multline}
	with positive constants $k_1$ and $k_2$ depending on $\alpha>-1$ and $\gamma>0$,
	\begin{equation}\label{k1-k2}
	\begin{split}
	k_1 &=2^{\frac{\gamma+3}{2}} \, \Gamma\!(\alpha +2 )\,\Gamma\!\left(\alpha+1\right)\, \Gamma\!\left(\frac{\gamma+3}{2}\right)\,\Gamma\left(\frac{\gamma+5}{2}\right),\\
	k_2 &= \frac{3 \sqrt{2}}{4}	\pi\left( 2\,\alpha +\frac{\gamma}{2}+2\right) \, \Gamma\!\left(\alpha + \frac{\gamma}{2}+1\right)^2.
	\end{split}
	\end{equation}
\end{proposition}
\begin{proof} The proof easily follows by plugging distribution function $\hat{f}^6$ into the definitions of non-convective fluxes \eqref{6 fields non-conv flux}. Indeed, for $f=\hat{f}^6$ we obtain
\begin{equation*}
p_{ij}= (p+\Pi) \delta_{ij}, \quad \sum_{i=1}^3 p_{iij}=0, \quad q_j =0, \quad i,j=1,2,3.
\end{equation*}
Details of the calculation of the production term \eqref{productionPii:equ} is given in Appendix  \ref{App production 6}.
\end{proof}

\begin{proposition}[Entropy law for six fields model]
The physical entropy density and its flux for the distribution function $\hat{f}^6$ from \eqref{f six} are
\begin{multline}\label{h6}
h(\hat{f}^6) = -k\, \frac{\rho}{m} \left\{ -\left( \alpha +\frac{5}{2}\right) + \right. \\ \left.
\log \left( \frac{\rho}{m} \left( \frac{m}{2 \pi k T} \right)^{\frac{3}{2}} \frac{1}{\left( 1 + \frac{\Pi}{p} \right)^{\frac{3}{2}}} \frac{1}{\Gamma(\alpha + 1)} \frac{1}{\left(k T\right)^{\alpha+1}} \left(\frac{1}{1-\frac{3}{2} \frac{\Pi}{(\alpha+1)p}}\right)^{\alpha+1} \right) \right\},
\end{multline}
and 
\begin{equation}\label{hj6}
h_j(\hat{f}^6) = U_j h(\hat{f}^6), \quad j=1,2,3.
\end{equation}
Moreover, the entropy density production term for the   cross section \eqref{model 3} reads
\begin{equation}\label{Sigma prop}
\begin{split}
\Sigma = - \frac{k \rho}{2  m p }  \left(  1- \frac{3}{2 (\alpha + 1)} \frac{\Pi}{p} \right)^{-1}\left(  1 + \frac{\Pi}{p} \right)^{-1} \frac{\left(\alpha + \frac{5}{2}\right)}{\left(\alpha+1\right)} \frac{\Pi}{p} \, \mathcal{P}, 
\end{split}
\end{equation}
with $\mathcal{P}$ calculated in \eqref{P iz Prop}, and the non-negativity 
$$\Sigma\geq 0$$ holds for all $\frac{\Pi}{p} $ in the range \eqref{range Pi/p}.
 \end{proposition}

\begin{proof}	
Plugging the distribution function \eqref{f six} into the definition of the entropy density \eqref{entripy phy} and its flux \eqref{entropy flux prod} we obtain \eqref{h6}--\eqref{hj6}. The production term $\Sigma$ is by virtue of \eqref{entropy flux prod} proved in the Appendix \ref{App production 6 entr}. Using the calculated $\mathcal{P}$ from \eqref{P iz Prop},  \eqref{Sigma prop} becomes
\begin{equation*}
\Sigma=  \frac{k}{2m}  \left(  1- \frac{3}{2 (\alpha + 1)} \frac{\Pi}{p} \right)^{-1}\left(  1 + \frac{\Pi}{p} \right)^{-1} \frac{\left(\alpha +\frac{5}{2}\right)^2}{\left(\alpha + 1\right)^2} C_{\mathcal{P}} \left(\frac{\Pi}{p} \right)^2 \geq 0,
\end{equation*}
for every $\Pi/p$ in the range of the validity of the model \eqref{range Pi/p}. 
\end{proof}

\subsection{Comparison with extended thermodynamics} \label{Sec ET}
In extended thermodynamics for six moments \cite{Rugg-Arima-Sugiyama-6-fields}, the non-equilibrium part of the entropy density was denoted with $\mathcal {K}$, for which the two conditions are prescribed: $(i)$ it vanishes for $\Pi=0$, and $(ii)$ it satisfies partial differential equation  (25) from \cite{Rugg-Arima-Sugiyama-6-fields}. Then applying the  entropy principle it was shown that the entropy production term $\Sigma$ is related to the production term $\mathcal{P}$ and the function $\mathcal {K}$ in the following way
\begin{equation}\label{residual}
\Sigma = \frac{1}{3} \frac{\partial \mathcal {K}}{\partial \Pi} \mathcal{P} >0.
\end{equation}

In this paper, starting from the kinetic theory we will determine the function $\mathcal {K}$, show that it satisfies PDE (25) from \cite{Rugg-Arima-Sugiyama-6-fields}, and then prove the above residual inequality \eqref{residual} for the production term $\mathcal{P}$ from  \eqref{P iz Prop}.   Moreover, the entropy production term of the form \eqref{residual} is equal to the one already calculated in \eqref{Sigma prop}. This  ensures compatibility of our kinetic, as much as macroscopic six fields model, with the extended thermodynamics.\\ 

Starting from the kinetic theory, one  solution of the PDE (25) from \cite{Rugg-Arima-Sugiyama-6-fields} can be found. As pointed out in \cite{Rugger-MEP, MPC-Madj-Sim},   function $\mathcal {K}$ can be a   difference of the entropy density evaluated at the distribution function $\hat{f}_6$ corresponding to the six fields problem   given in \eqref{f six}  and the local  equilibrium distribution function $f_M$ from \eqref{Maks}. In our notation, 
\begin{equation*}
\mathcal {K}\left(\rho, p, \Pi\right) := h(\hat{f}^6) - h(f_M).
\end{equation*}
 Using expressions \eqref{h6},  \eqref{Maks} and definition of the physical entropy \eqref{entripy phy}, we obtain the following form of $\mathcal {K}$,
\begin{equation*}
\label{kappa6}
\mathcal {K}(\rho, p, \Pi)= \frac{k \rho}{m} \log \left\{ \left(1+\frac{\Pi}{p}\right)^{\frac{3}{2}}\left(1-\frac{3 \Pi}{2(\alpha+1)p}\right)^{\alpha +1 }\right\}.
\end{equation*} 
Clearly, for $\Pi=0$ we get $ \mathcal {K}(\rho, p, 0)=0$.
It is also straightforward to prove that such  $\mathcal {K}$ satisfies PDE from \cite{Rugg-Arima-Sugiyama-6-fields}, which in our notation reads
\begin{equation*}
	\rho \, \frac{\partial \mathcal {K}}{\partial \rho} + \left(\frac{p + \Pi}{\left(\alpha +\frac{5}{2}\right)}+p\right) \frac{\partial \mathcal {K}}{\partial p}+ \left\{ (p + \Pi)\left(\frac{5}{3}-\frac{1}{\left( \alpha+\frac{5}{2} \right)}\right)-p\right\} \frac{\partial \mathcal {K}}{\partial \Pi}-\mathcal {K} + \frac{\Pi}{T}=0.
\end{equation*}
Then we can calculate the derivative
\begin{equation}\label{k der}
\frac{\partial \mathcal {K}}{\partial \Pi} = -\frac{3}{2 } \frac{k \rho}{m p}  \frac{\left(\alpha+\frac{5}{2}\right)}{\left(\alpha+1\right)}   \left(  1- \frac{3}{2 (\alpha + 1)} \frac{\Pi}{p} \right)^{-1}\left(  1 + \frac{\Pi}{p} \right)^{-1} \frac{\Pi}{p}. 
\end{equation} 
With the expression above, it is easy to see that the entropy production \eqref{residual} coincides with \eqref{Sigma prop}.

Therefore, we have proven that our kinetic model provides the six fields model compatible with the entropy principle from extended thermodynamics.

\subsection{Relaxation time} The relaxation time $\tau_{\Pi}$ for the non-equilibrium variable $\Pi$ is obtained by  linearizing the production term $\mathcal{P}$ around $\Pi=0$, that yields
\begin{equation*}
\overline{\mathcal{P}}=-\frac{1}{\tau_{\Pi}}  \Pi,
\end{equation*}
for 
	\begin{equation*}
	\tau_{\Pi}=\left\{     \frac{\rho}{m} \left(\frac{p}{\rho}\right)^{\frac{\gamma}{2}} \left\| b \right\|_{L^1(\mathrm{d}\sigma)}  \frac{\left(\alpha+\frac52\right)}{\left(\alpha+1\right)}\sqrt{\frac{2}{\pi}} \frac{1}{\Gamma\left( \frac{4\alpha+\gamma+9}{2} \right)}\left(2^{\frac{\gamma}{2}} k_1+ k_2 \right)  \right\}^{-1},
	\end{equation*}
	where constants $k_1$ and $k_2$ are from \eqref{k1-k2}.

\section{Fourteen moments model}\label{Sec: macro 14}

The macroscopic model of fourteen moments is obtained by extending the list of test functions given in \eqref{test 6} that will allow to take into account evolution equations for momentum  and heat  fluxes. Namely, instead of \eqref{test 6} we  consider the following test functions
\begin{equation}\label{test 14}
m, \ mv, \ m \, v_i \, v_j, \ \frac{m}{2} \left| v\right|^2 + I, \left(\frac{m}{2} \left| v\right|^2 + I\right)v_j.
\end{equation}
If addition to \eqref{6 densities} and \eqref{6 fields non-conv flux} we define the non-convective fluxes
\begin{equation}\label{14 densities}
\left( \begin{array}{c}
p_{ijk} \\
q_{ij}
\end{array}\right)   =  \int_{\mathbb{R}^3 \times \mathbb{R}_+}
\left(
\begin{array}{c}
m\, c_i\, c_j \, c_k\\
\left(\frac{m}{2} \left| c\right|^2 + I\right)c_ic_j
\end{array}
\right)\,
f \, \mathrm{d} I \, \mathrm{d}{v},
\end{equation}
where  $c=v-U$. 
Then integration of the Boltzmann equation \eqref{BE non-weight} against the test functions \eqref{test 14} yields system of equations governing the 14 moments, namely,
\begin{align*}
& \partial_{t} \rho + \sum_{j=1}^3\partial_{x_j} (\rho \, U_{j}) = 0,
\nonumber \\
& \partial_{t} (\rho\, U_{i}) +  \sum_{j=1}^3 \partial_{x_j}
(\rho\, U_{i} \, U_{j} + p_{ij}) = 0,
\\
& \partial_{t} \left( \rho \, U_{i}\, U_{j} + p_{ij} \right)
+\sum_{k=1}^3  \partial_{x_k} \left\{ \rho \, U_{i} \, U_{j}\, U_{k}
+ U_{i} p_{jk} + U_{j} p_{ki} + U_{k} p_{ij} + p_{ijk}\right\}  =P_{ij},
\nonumber
\end{align*}
\begin{align}
& \partial_{t} \left( \frac{1}{2} \rho \left|U\right|^{2} +
\rho e \right) +  \sum_{i=1}^3 \partial_{x_i} \left\{ \left(
\frac{1}{2} \rho \left| U \right|^{2} + \rho e \right)
U_{i} +\sum_{j=1}^3 p_{ij} U_{j} + q_{i} \right\} = 0,
\label{14 mom non closed} \\
&
	 \partial_{t} \left\{ \left(
\frac{1}{2} \rho \left|U\right|^{2} + \rho e \right)
U_{i} + \sum_{j=1}^3 p_{ij} U_{j} + q_{i} \right\}
+  \sum_{j=1}^3 \partial_{x_j} \left\{ \left( \frac{1}{2} \rho \left|U\right|^{2}
+ \rho e \right) U_{i} U_{j} \right.
\nonumber
\\
& \left. +\sum_{k=1}^3 \left( U_{i} U_{k} p_{jk} + U_{j} U_{k} p_{ik}\right)
+ \frac{1}{2} \rho \left|U\right|^{2} p_{ij}+ \sum_{k=1}^{3}U_k p_{ijk}+q_iU_j + q_j Ui +q_{ij} \right\} = Q_i.      
\nonumber
\end{align}
with the production terms
\begin{equation}\label{14 prod}
\left( \begin{array}{c}
P_{ij} \\
Q_{i}
\end{array}\right)   =  \int_{\mathbb{R}^3 \times \mathbb{R}_+}
\left(
\begin{array}{c}
m\, v_i\, v_j\\
\left(\frac{m}{2} \left| v\right|^2 + I\right) v_i
\end{array}
\right)\,
Q^{nw}(f,f)(v,I)  \, \mathrm{d} I \, \mathrm{d}{v},
\end{equation}
for any $i,j=1,2,3$. 
Our aim is to close the system above in an approximative setting by exploiting the maximum entropy principle and then proceeding with the  appropriate linearization of the distribution function around the local equilibrium. Production terms are linearized as well and calculated for the cross section \eqref{model 3}. Linearization is unavoidable since the exact solution of the variational problem does not yield convergent moments.

\begin{lemma}[Fourteen moments distribution function]\label{Prop: f14}
The solution of  the maximum entropy principle
\begin{alignat*}{3}
\max_f& \quad h = - k   \int_{\mathbb{R}^3 \times [0,\infty)} f \, \log(f I^{-\alpha}) \, \mathrm{d}I \, \mathrm{d}v \nonumber\\
\text{s.t.}  & \left( \begin{array}{c}
\rho \\
0_i \\
\left( \alpha+ \frac{5}{2}\right)p\\
p_{ij}\\
q_i
\end{array}\right)   =  \int_{\mathbb{R}^3 \times \mathbb{R}_+}
\left(
\begin{array}{c}
m \\
m\, c_{i} \\
\frac{m}{2} \left| {c} \right|^2 + I\\
m \, c_{i} c_{j} \\
\left(\frac{m}{2} \left| {c} \right|^2 + I \right) c_i 
\end{array}
\right)\,
f \, \mathrm{d} I \, \mathrm{d}{c},
\end{alignat*}
where $\alpha>-1$ and $\sum_{i=1} ^3 p_{ii}=3(p+\Pi)$, linearized around the local equilibrium state $f_M$ defined in  \eqref{Maks} is given with
\begin{multline}
\label{f fourteen}
\hat{f}^{14}\approx {f}^{14}  = f_M \left\{1-\frac{\rho}{p^2} q \cdot c -\frac{3}{2\left(\alpha+1\right)}\frac{\Pi \rho}{m p^2}\left(\frac{m}{2}\vert c \vert^2+I\right) \right.\\
\left.+\frac{\rho}{2p^2} \sum_{i,j=1}^3 \left( p_{\langle ij \rangle}+\frac{\left(\alpha + \frac{5}{2}\right)}{\left(\alpha+1\right)} \delta_{ij} \Pi \right)c_i c_j +\left(\alpha +\frac{7}{2}\right)^{-1} \frac{\rho^2}{m p^3} q\cdot c \left(\frac{m}{2} \left| {c} \right|^2 + I \right) \right\},
\end{multline}
where the following notation is used 
\begin{equation*}
 p_{\langle ij \rangle} = p_{ij} - \delta_{i j} \frac{1}{3} \sum_{\ell=1} ^3 p_{\ell \ell} =  p_{ij} - (p+\Pi) \delta_{i j}. 
\end{equation*}
\end{lemma}
The proof of this lemma follows the one given in \cite{MPC-Rugg-Sim}, by Remark \ref{Rem}. \\

The distribution function \eqref{f fourteen} enables to close the system of equations corresponding to fourteen moments \eqref{14 mom non closed} in the linearized form, for the chosen cross section as in \eqref{model 3} under an additional assumption of the constant angular function $b\big(\tfrac{u}{\left| u \right|}\cdot\sigma\big)$. Namely, the following proposition holds.
\begin{proposition}[Fourteen fields system of equations]\label{Prop 14 fields} Closed system of equations for fourteen moments reads
	\begin{align*}
	& \partial_{t} \rho + \sum_{j=1}^3\partial_{x_j} (\rho \, U_{j}) = 0,
	\nonumber \\
	& \partial_{t} (\rho\, U_{i}) +  \sum_{j=1}^3 \partial_{x_j}
	(\rho\, U_{i} \, U_{j} + p_{ij}) = 0,
	\\
	& \partial_{t} \left( \rho \, U_{i}\, U_{j} + p_{ij} \right)
	+\sum_{k=1}^3  \partial_{x_k} \left\{ \rho \, U_{i} \, U_{j}\, U_{k}
	+ U_{i} p_{jk} + U_{j} p_{ki} + U_{k} p_{ij} \right.
	\nonumber \\
	&  \quad \left. + \left( \alpha+\frac{7}{2} \right)^{-1}
	\left( q_{i} \delta_{jk} + q_{j} \delta_{ki}
	+ q_{k} \delta_{ij} \right) \right\}  =P_{ij},
	\nonumber
	\end{align*}
	\begin{align*}
	& \partial_{t} \left( \frac{1}{2} \rho \left|U\right|^{2} +
	\rho e \right) +  \sum_{i=1}^3 \partial_{x_i} \left\{ \left(
	\frac{1}{2} \rho \left| U \right|^{2} + \rho e \right)
	U_{i} +\sum_{j=1}^3 p_{ij} U_{j} + q_{i} \right\} = 0,
	\nonumber \\
	& \partial_{t} \left\{ \left(
	\frac{1}{2} \rho \left|U\right|^{2} + \rho e \right)
	U_{i} + \sum_{j=1}^3 p_{ij} U_{j} + q_{i} \right\}
	+  \sum_{j=1}^3 \partial_{x_j} \left\{ \left( \frac{1}{2} \rho \left|U\right|^{2}
	+ \rho e \right) U_{i} U_{j} \right.
	\\
	&\quad \left. +\sum_{k=1}^3 \left( U_{i} U_{k} p_{jk} + U_{j} U_{k} p_{ik}\right)
	+ \frac{1}{2} \rho \left|U\right|^{2} p_{ij} 
	+ \left( \alpha +\frac{9}{2}\right)\left( \alpha+\frac{7}{2}\right)^{-1}\left(q_i\, U_j +q_j\, U_i \right)\right. \\
	&\quad \left.+ \left( \alpha +\frac{7}{2} \right)^{-1}  \delta_{i j} \sum_{k=1}^3 q_k U_k+ \left(  \alpha+\frac{9}{2} \right) \frac{p}{\rho} \, p_{ij}
	- \frac{p^{2}}{\rho} \, \delta_{ij} \right\} = Q_i,  
	\nonumber
	\end{align*}
	for $i,j=1,2,3$, where we have assumed the  relations \eqref{polytrop}. The production terms for the cross section \eqref{model 3}  with $b\!\,\big(\tfrac{u}{\left| u \right|} \cdot \sigma\big) =K$, $K$ is a constant, linearized around the global equilibrium state \eqref{Maks} read
	 \begin{multline} \label{PijFINAL}
	 \overline P_{ij}^{14}=-K \frac{\rho}{m} \left(\frac{p}{\rho}\right)^{\frac{\gamma}{2}}  \frac{\sqrt{\pi}}{\Gamma\!\left(\frac{4\alpha+\gamma+9}{2}\right)} \left\{\frac{1}{15} (4\alpha+\gamma+7)  \left(2^{\gamma+2} (\gamma+5) n_1 + 15 n_2  \right) p_{\langle ij \rangle}  \right. \\
	 \left. +  \left(\alpha+\frac{5}{2}\right) \left( \frac{2^{\gamma+4}}{3}  n_1 + \frac{(4\alpha +\gamma+4)}{(\alpha+1)}n_2\right)  \Pi\, \delta_{ij} \right\},
	 \end{multline}
	 \begin{multline}\label{QiFINAL}
	 \overline Q_{i}^{14}=\sum_{k=1}^3 U_k\, \overline P_{ki}^{14}-K \frac{\rho}{m}  \left(\frac{p}{\rho}\right)^{\frac{\gamma}{2}} q_i     \left( \alpha +\frac{7}{2}\right)^{-1}  \frac{\sqrt{\pi}}{72 \, \Gamma\!\left(\frac{4\alpha+\gamma+9}{2}\right)}  \\
	 \times  \left(2^{\gamma+5}((4\alpha+\gamma)(3\alpha+\gamma)+57\alpha+15\gamma+60) n_1 \right. \\
	 \left.+9((4\alpha+\gamma)(2(4\alpha+\gamma)+\gamma^2+38)+7\gamma^2+160) n_2\right),
	 \end{multline}
	 with the positive constants depending on $\alpha>-1$ and $\gamma>0$,
	 \begin{equation}\label{const n}
	 n_1 = \Gamma\!\left(\alpha+1\right)^2 \Gamma \!\left(\frac{\gamma+3}{2}\right)\Gamma \!\left(\frac{\gamma+5}{2}\right), \quad
	 n_2 = \pi \Gamma\! \left(\alpha +\frac{\gamma}{2}+1\right)^2.
	 \end{equation}
\end{proposition}
\begin{proof}
The non-convective fluxes are obtained by plugging the distribution function \eqref{f fourteen} into their definition \eqref{14 densities},
\begin{equation*}
\begin{split}
p_{ijk} & =  \left( \alpha+\frac{7}{2} \right)^{-1}
\left( q_{i} \delta_{jk} + q_{j} \delta_{ki}
+ q_{k} \delta_{ij} \right),\\
q_{ij} &=  \left(  \alpha+\frac{9}{2} \right) \frac{p}{\rho} \, p_{ij}
-\frac{p^{2}}{\rho} \, \delta_{ij},
\end{split}
\end{equation*}
while the production terms are calculated in the Appendix \ref{App production14 mom}.
\end{proof}

\subsection{Relaxation times and transport coefficients}
In sense of extended thermodynamics and the theory of hyperbolic systems of balance laws \cite{Rugg-Poly}, production terms can be represented in the following form
\begin{equation*}
\overline P_{ij}^{14}=-\frac{1}{\tau_s} p_{\langle ij \rangle}- \frac{1}{\tau_{\Pi}} \Pi \delta_{ij}, \quad \overline Q_{i}^{14}=\sum_{k=1}^3 U_k \overline P_{ki}^{14}-\frac{1}{\tau_q} q_i,
\end{equation*}
where $\tau_s, \tau_{\Pi}, \tau_q$ are appropriate relaxation times.  It is also known that relaxation times can be related to the transport coefficients - shear viscosity $\mu$, bulk viscosity $\nu$, and heat conductivity $\kappa$ in the following manner,    
\begin{equation*}
\mu=p \, \tau_s, \quad \nu =\frac{4(\alpha+1)}{3(2\alpha+5)} \, p \, \tau_{\Pi}, \quad \kappa = \left(\alpha+\frac{7}{2}\right)\frac{p^2}{\rho T} \, \tau_q .
\end{equation*}
In extended thermodynamics those parameters are of phenomenological nature. Starting from the Boltzmann equation,  calculation of the production terms allows to obtain their explicit expressions, that will depend on $\alpha>-1$ related to the number of internal degrees of freedom and potential $\gamma>0$ from the cross section \eqref{model 3}. More precisely,  from \eqref{PijFINAL} and \eqref{QiFINAL} it can be easily  recognized
\begin{multline}\label{visco}
\mu=\frac{m }{K} \left( \frac{p}{\rho} \right)^{1-\frac{\gamma}{2}}   \frac{\Gamma\!\left(\frac{4\alpha+\gamma+9}{2}\right)}{\sqrt {\pi}} 15
\left( (4\alpha+\gamma+7)\left(2^{\gamma+2} (\gamma+5) n_1+15\, n_2  \right)\right)^{-1},
\end{multline}
\begin{multline*}
 \nu =   \frac{m }{K} \left( \frac{p}{\rho} \right)^{1-\frac{\gamma}{2}}  \frac{2(\alpha+1)^2}{3(\alpha+\frac{5}{2})^2}  \frac{\Gamma\!\left(\frac{4\alpha+\gamma+9}{2}\right)}{\sqrt {\pi}}\\
\times \left(\frac{2^{\gamma+4}}{3} (\alpha+1)  n_1+(4\alpha +\gamma+4) n_2 \right)^{-1},
\end{multline*}
\begin{multline*}
\kappa=  \frac{k }{K}  \left( \frac{p}{\rho} \right)^{1-\frac{\gamma}{2}}  \left(\alpha+\frac{7}{2}\right)^2 \frac{72 \Gamma\!\left(\frac{4\alpha+\gamma+9}{2}\right)}{\sqrt{\pi}} \\
\times  \left(2^{\gamma+5}((4\alpha+\gamma)(3\alpha+\gamma)+57\alpha+15\gamma+60) n_1 \right. \\
	 \left.+9((4\alpha+\gamma)(2(4\alpha+\gamma)+\gamma^2+38)+7\gamma^2+160) n_2\right)^{-1},
\end{multline*}
where $n_1$ and $n_2$ are from \eqref{const n}.

\subsection{Prandtl number} One of   tests for the validity of the transport coefficients given above is to verify the value of the Prandtl number, defined in our notation as
\begin{equation}\label{Pr tr coef}
\text{Pr}=\left(\alpha+\frac{7}{2}\right) \frac{k}{m}\frac{\mu}{\kappa}.
\end{equation}
From the other side, theoretical value of the Prandtl number for polyatomic gases can be obtained by Eucken's relation that in our notation reads
\begin{equation}\label{Pr al}
\text{Pr} = \frac{4\alpha  + 14}{4 \alpha + 19}.
\end{equation}
The goal is to find a $\gamma>0$ such that for certain values of $\alpha>-1$ the two expressions \eqref{Pr tr coef} and \eqref{Pr al} are equal. 

The value of $\alpha$ is related to  modes of a polyatomic molecule, as shows   Table \ref{Tab}.  

\begin{center}
	\begin{table}[h]
		\caption{Number of degrees of freedom $D$ for different modes (combinations of translation/rotation/vibration) where $\mathcal{N} \geq 2$  is the number of atoms in a polyatomic molecule, with the corresponding value of $\alpha=\frac{D-5}{2}$,  theoretical value of the Prandtl number from \eqref{Pr al}  and the value of $\gamma$ enabling that this theoretical value of the Prandtl number coincides with the one given in \eqref{Pr tr coef}, i.e. enabling that the two expressions  \eqref{Pr tr coef} and \eqref{Pr al} are equal.}
	\begin{tabular}{| m{3.5cm}||m{1.6cm}|| m{1.6cm}||m{2cm}| @{}m{0cm}@{} } \hline
		& \multicolumn{2}{c||}{Translation and rotation} & \multirow{2}{2cm}{Translation, rotation and vibration} & \\[5pt] \cline{2-3}
		 & Linear molecule &  Non-linear molecule & & \\[5pt] \hline
		 \hfil  Degrees of freedom   & \hfil 5  & \hfil6 & \hfil$3\mathcal{N}$ & \\[5pt] \hline
		\hfil 	  $\alpha$  & \hfil0 &\hfil $\frac{1}{2}$ &\hfil $\frac{1}{2}(3 \mathcal{N}-5) $ &   \\[5pt]  \hline
	\hfil	   $\text{Pr}$ from \eqref{Pr al}   & \hfil$\frac{14}{19}$ &\hfil $\frac{16}{21}$ & \hfil$\frac{6\mathcal{N}+4}{6\mathcal{N}+9}$& \\[5pt] \hline
		\hfil    $\gamma$   & \hfil$2.153$ &\hfil $2.368$ & \hfil Table \ref{Tab 2}& \\[5pt] \hline
	\end{tabular}
\label{Tab}
\end{table}
\end{center}

As it can be seen in Table \ref{Tab}, for $\alpha=0$ the theoretical value of the Prantdtl number $14/19$ is obtained from \eqref{Pr tr coef} by taking $\gamma=2.153$. Considering $\alpha=\frac{1}{2}$,  the value $\gamma=2.368$ in \eqref{Pr tr coef} recovers $\text{Pr} = \frac{16}{21}$. When vibrational modes are also taken into account, for any number of atoms $\mathcal{N}\geq 2$ we can find the value $\gamma>0$ such that the correct value of the Prandtl number \eqref{Pr al} is obtained, as shows Table \ref{Tab 2}.

\begin{center}
	\begin{table}[h]
		\caption{ The number  $\mathcal{N}$ of atoms in  a   polyatomic molecule and the corresponding value of potential $\gamma$ such that the  theoretical value of the Prandtl number from \eqref{Pr al} is equal to the one in \eqref{Pr tr coef}.}
		\begin{tabular}{ |c | c | c | c | c | c | c |c | c|  } \hline
	$\mathcal{N}$ & 3 & 4 & 5 & 6& 7& 8 & 9 & 10 \\ \hline 
	$\gamma$ & 4.063 & 9.469 & 17.262&25.801 & 34.705 & 43.835 &53.123&62.526\\ \hline
		\end{tabular}
	\label{Tab 2}
	\end{table}

\end{center}

\subsection{Dependence of the shear viscosity on temperature} Another physical validity of the proposed model can be provided by studying  temperature dependence of the shear viscosity \eqref{visco}. Our goal is to compare the shear viscosity \eqref{visco}  with experimental data given in \cite{Cha-Cow} for  the room temperature range 293-373K, and in \cite{Exper-2, Exper-1} for high temperatures in the range 600-2000K.

The shear viscosity issuing from the kinetic theory \eqref{visco},  provides the following dependence on temperature,
\begin{equation}\label{visco expo}
\mu \sim T^{1-\frac{\gamma}{2}},
\end{equation}
where $\gamma$ is related to the choice of the cross-section \eqref{model 3} with constant angular part, as stated in Proposition \ref{Prop 14 fields}. We point out  that such a relation makes sense only if $\gamma<2$, since  it is observed that shear viscosity of gases increases as temperature grows \cite{Cha-Cow}.\\

In \cite{Cha-Cow}  the following relation  is assumed,
\begin{equation}\label{visco exp s}
\mu \sim T^s,
\end{equation}
and in 	Table 14, page 232,    experimental  values for $s$ on the temperature range 293-373 K are provided. Direct comparison of \eqref{visco expo} and \eqref{visco exp s}  gives the relation between $s$ and $\gamma$,
\begin{equation}\label{s-gamma}
\gamma=-2s+2.
\end{equation}
Note that the comparison is only possible for $s<1$, in order to ensure positivity of $\gamma$. \\

For different polyatomic molecules, the goal is to  adjust the value of $\gamma$ in order to match the experimentally measured $s$ related to  $\gamma$ by virtue of \eqref{s-gamma}. Combining this $\gamma$ with $\alpha$ coming from the structure of a molecule  gives the value of the Prantdl number using \eqref{Pr tr coef}. That value can be compared to the theoretical one obtained in \eqref{Pr al}. Table \ref{Table 3} shows the results.

\begin{center}
	\begin{table}[h]
		\caption{Experimental values of $s$ \cite{Cha-Cow} for different molecules revealing the  dependence of the shear viscosity upon  temperature $\mu \sim T^s$ given in \eqref{visco exp s}, the corresponding  value of $\gamma$ by virtue of \eqref{s-gamma},  and the Prandtl number from \eqref{Pr tr coef}. This value of the Prandtl number is further compared to the theoretical one \eqref{Pr al} and the relative error is provided.}
\begin{tabular}{ |c  || c || c || c || c|| c |} \hline
			Gas & $s$ &  $\gamma$ &  Pr from \eqref{Pr tr coef} & Pr from \eqref{Pr al} & Relative error \\ \hline 
			H$_2$ & 0.668  & 0.664 & 0.816 & 0.737 & 10.7\%   \\ \hline 
			CO & 0.734  & 0.532 & 0.819 & 0.737 &  11.1\% \\ \hline 
			N$_2$ & 0.738  & 0.524 & 0.819 & 0.737 & 11.1\%   \\ \hline 
			NO & 0.788 & 0.424 & 0.82 & 0.737 &  11.3\% \\ \hline 
			O$_2$ & 0.773  & 0.454 & 0.82 & 0.737 &  11.3\% \\ \hline 
			CO$_2$ & 0.933 & 0.134 & 0.819 & 0.737 & 11.1\%   \\ \hline 
			N$_2$O & 0.943 & 0.114 & 0.819 & 0.737 & 11.1\%   \\ \hline 
			CH$_4$ & 0.836  & 0.328 & 0.849 & 0.762 & 10.3\%   \\ \hline 
		\end{tabular}
	\label{Table 3}
	\end{table}
\end{center}

On the other side, 	for higher temperatures, we have to consider vibrational modes as well.  In \cite{Exper-1, Exper-2} experimental data for pointwise  values of the shear viscosity at certain high temperatures in the range 100-2000 K can be found. We first fit those data  in the manner of  \eqref{visco expo}, i.e. we find coefficients $A$ and $s$ such that $\mu= A \, T^s$, as illustrated at the Figure \ref{figure-data}.
\begin{figure}[t]
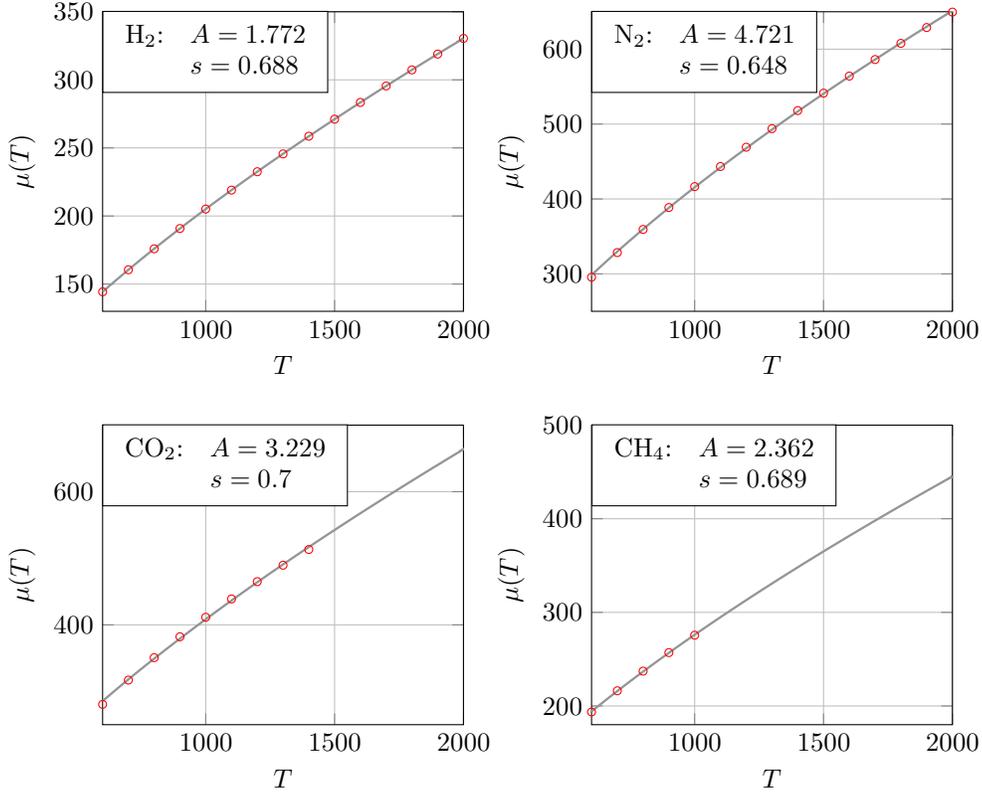

	\begin{center}
		%
	
	\end{center}
	\caption{Shear viscosity as a function of the temperature in the form $\mu(T)=A T^s$. Points on the plot are experimentally observed values \cite{Exper-1, Exper-2}, while  solid lines represent fitted curves.  }
	\label{figure-data}
\end{figure}
The fitted value of $s$  leads to the corresponding  value of $\gamma$ for which we compute Prandtl number and compare it to the value given by Eucken's relation. The results are shown in   Table \ref{Tab 4}.

\begin{center}
	\begin{table}[h]
		\caption{ Value of experimental observed parameter $s$ \cite{Exper-1,Exper-2}  and the corresponding value of $\gamma$ by means of \eqref{s-gamma},  expressing the  dependence of shear viscosity $\mu$ of the shape \eqref{visco expo} upon high temperature for different molecules and the Prandtl number from \eqref{Pr tr coef}. This value of the Prandtl number is further compared to the theoretical one \eqref{Pr al} and the relative error is provided.}
		\begin{tabular}{ |c || c || c || c || c|| c |} \hline
			Gas & $s$ & $\gamma$ &  Pr from \eqref{Pr tr coef} & Pr from \eqref{Pr al} & Relative error \\ \hline 
			H$_2$ & 0.688  & 0.624 & 0.847 & 0.762 & 11.2\%   \\ \hline 
			N$_2$ & 0.684 & 0.704 & 0.846 & 0.762 &  11.0\% \\ \hline 
			CO$_2$ & 0.7  & 0.599& 0.894 & 0.815& 9.7\%   \\ \hline 
			CH$_4$ & 0.689 & 0.419 & 0.930 & 0.872& 6.8\%   \\ \hline 
		\end{tabular}
	\label{Tab 4}
	\end{table}
\end{center}

We conclude that for a fixed polyatomic molecule, i.e. fixing $\alpha$, the same value of  $\gamma$ can provide agreement with the experimental data concerning dependence of shear viscosity on temperature and the value of the Prandtl number, which coincides with the theoretical one given by Eucken's relation \eqref{Pr al} at a relative error ranging from 6.8 - 11.3\%.  These results are valid at any temperature range as long as the power in \eqref{visco exp s} is less than one.\\

We mention that the similar analysis was performed in \cite{MPC-Sim} for the kinetic model in the non-weighted setting and the model for the cross-section containing one free parameter that could be matched in order to recover the correct  temperature dependence of the viscosity, yet  only for CO,  yielding $\textrm{Pr}=0.781$, which is in an satisfactory agreement with the theoretical value for diatomic gases $\textrm{Pr} = 0.737$ from Table \ref{Table 3}. We also remark that in \cite{MPC-Sim} the exponent of temperature in the experimental viscosity relation \eqref{visco exp s} depends on $\alpha$ which is not the case here as it can be seen in \eqref{visco expo}, because of the additional term in the collision operator weak form present in the non-weighted setting that  involves $I^\alpha I_*^\alpha$ which subtracts dependence on $\alpha$.  Therefore,  our results improve the ones from \cite{MPC-Sim}, since the described analysis applies to  all gases for which $s<1$ in \eqref{visco exp s}, using experimental data \cite{Cha-Cow, Exper-1, Exper-2}.  \\ 
  
Finally, we put in evidence the key of success of this analysis. We first introduce the difference of the two expressions for the Prandtl number \eqref{Pr tr coef} and \eqref{Pr al},
\begin{equation*}
\Delta(\gamma, \alpha)= \left(\alpha+\frac{7}{2}\right) \frac{k}{m}\frac{\mu}{\kappa}  -  \frac{4\alpha  + 14}{4 \alpha + 19}.
\end{equation*}
 Tables \ref{Tab} and \ref{Tab 2} show the values $\gamma^*$ such that $\Delta(\gamma^*, \alpha)=0$ for  the fixed $\alpha$. It can be observed that  $\gamma^*$ grows with the increase of $\alpha$. However, the  analysis of viscosity  dependence on temperature requires $\gamma<2$ and at the same time provides satisfactory agreement with the value of the Prandtl number given by Eucken's relation \eqref{Pr al}. The reason is that the cross-section model \eqref{model 3} used in this paper yields that, for any fixed value of $\alpha>-1$,  the difference  $\Delta(\gamma,\alpha)$ is close to zero when $\gamma \in(0,2)$, which is illustrated at  the Figure  \ref{figure-pr}. This difference can be possibly reduced  for an another model of the cross-section, or taking a different approximation of transport coefficients, for instance as in \cite{Str-Ra}.

\begin{figure}[t]
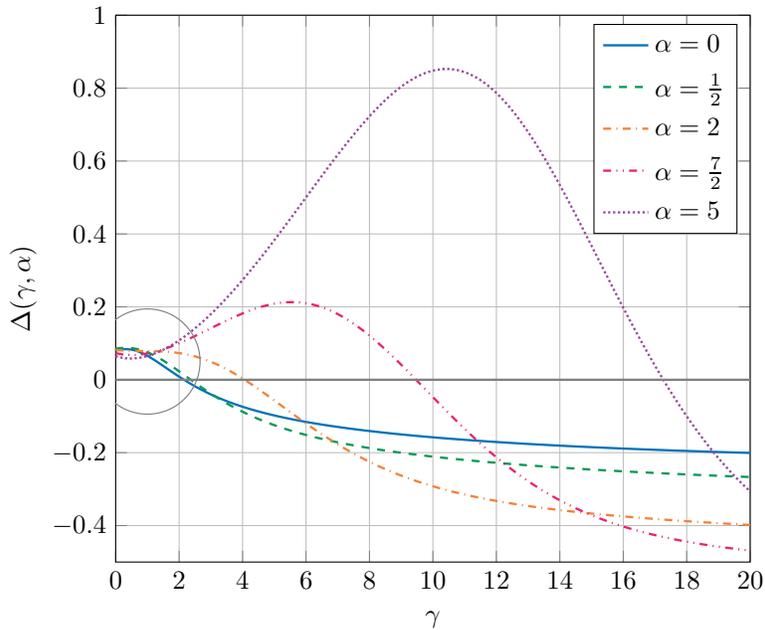

	\begin{center}
	
	\end{center}
	\caption{Dependence of $\Delta(\gamma,\alpha)$ in $\gamma$ for certain values of $\alpha$, with a particular emphasize on the behavior for $0<\gamma<2$.  }
	\label{figure-pr}
\end{figure}

\section{ Acknowledgments} 
The authors would like to thank Prof. Srboljub Simi\'c and Prof. Manuel Torrilhon  for fruitful discussions on the topic and to Prof. Thierry Magin for enlightening   physical aspects of polyatomic gas modelling. Authors thank to the Department of Mathematics and Informatics, Faculty of Sciences, University of Novi Sad for the hospitality. 
V. Djordji\'c would like to acknowledge that this publication is based upon work from COST Action CA18232 MAT-DYN-NET, supported by COST (European Cooperation in Science and Technology), through  the STSM (Short Term Scientific Mission) at the RWTH Aachen University. M. Pavi\'c-\v Coli\'c acknowledges the financial support of the Ministry of Education, Science and Technological Development of the Republic of Serbia (Grant No. 451-03-68/2020-14/ 200125), and the support of the Program for Excellent Projects of Young Researchers (PROMIS) of the Science Fund of the Republic of Serbia within the project MaKiPol \#6066089, \textit{Mathematical methods in the kinetic theory of polyatomic gas mixtures: modelling, analysis and computation}. N. Spasojevi\' c    acknowledges support from Oden Institute at the University of Texas Austin and funding by  DOE DE-SC0016283 project \textit{Simulation Center for Runaway Electron Avoidance and Mitigation}.

\appendix

\section{Computation of the collision frequency}\label{App coll fre}

For the sake of simplicity, the equilibrium distribution function \eqref{Maks} will be written as
\begin{equation}\label{Maxwell dmsless}
f_M=I^{\alpha} L_0  e^{-\frac{1}{k T} (\frac{m}{2} \vert c\vert^2 +I)} \quad \text{where} \quad L_0= \frac{\rho}{m} \left( \frac{m}{2 \pi k T} \right)^{\frac{3}{2}} \frac{1}{\Gamma(\alpha + 1)} \frac{1}{\left(k T\right)^{\alpha+1}}.
\end{equation}

Our goal is to compute the collision frequency defined in \eqref{coll freq}, evaluated for the equilibrium distribution function \eqref{Maxwell dmsless} and the  cross section \eqref{model 3},
\begin{multline}\label{coll freqA1}
\nu_{\gamma , \alpha}(v, I)= L_0  \int_{\mathbb{R}^3 \times [0,\infty) \times [0,1]^2 \times S^2} I_*^{\alpha}  e^{-\frac{1}{k T} (\frac{m}{2} \vert v_*-U\vert^2 +I_*)} \\	 \times b\!\,\big(\tfrac{v-v_*}{\left| v-v_* \right|} \cdot \sigma\big)  \left( R^{\frac{\gamma}{2}} |v-v_*|^\gamma  + \left( r (1-R)\frac{I}{m} \right)^{\frac{\gamma}{2}} + \left( (1-r) (1-R)\frac{I_*}{m} \right)^{\frac{\gamma}{2}} \right) \\
\times (r(1-r))^{\alpha} \, (1-R)^{2\alpha+1}  R^{\frac{1}{2}} \, \mathrm{d} \sigma \, \mathrm{d} r \, \mathrm{d} R\, \mathrm{d} I_* \, \mathrm{d}  v_*.
\end{multline}
For the constant issuing from the integration respect to $r$ and $R$, we introduce the following notation
\begin{multline}\label{const r R}
C_{(a,b,c)} = \int_{[0,1]^2} (r(1-r))^{ \alpha} \,  (1-R)^{2\alpha+1} R^{\frac{1}{2}} (1-R)^a R^b r^c  \mathrm{d} r\, \mathrm{d} R \\= \frac{\Gamma\!\left(  2\alpha+a +2\right)\Gamma\!\left(b+\frac{3}{2}\right) \Gamma\!\left(\alpha+c+1\right) \Gamma\!\left(\alpha+1\right) }{\Gamma\!\left( 2\alpha+a+b+\frac{7}{2} \right) \Gamma\!\left(2\alpha+c+2\right) },
\end{multline}
where $\Gamma$ stands for the Gamma function.
Then the collision frequency \eqref{coll freqA1} when integrated respect to $\sigma, r, R$ becomes
\begin{multline*}
\nu_{\gamma , \alpha}(v, I)=  L_0 \left\| b \right\|_{L^1(\mathrm{d}\sigma)} \int_{\mathbb{R}^3 \times [0,\infty)} I_*^{\alpha}  e^{-\frac{1}{k T} (\frac{m}{2} \vert v_*-U\vert^2 +I_*)}\\
\times \left( C_{(0, \frac{\gamma}{2},0)} |v-v_*|^\gamma +C_{(\frac{\gamma}{2},0, \frac{\gamma}{2})}\left( \left(\frac{I}{m}\right)^\frac{\gamma}{2}+\left(\frac{I_*}{m}\right)^\frac{\gamma}{2}\right) \right) \mathrm d I_* \; \mathrm d v_*.
\end{multline*}
Next, integration with  respect to $I_*$ leads to
\begin{multline}\label{pomocna coll}
\nu_{\gamma , \alpha}(v, I)=  \frac{\rho}{m} \left( \frac{m}{2 \pi k T} \right)^{\frac{3}{2}} \frac{ \left\| b \right\|_{L^1(\mathrm{d}\sigma)}}{m^\frac{\gamma}{2}} \int_{\mathbb{R}^3}  e^{-\frac{m}{2 k T}  \vert v_*-U\vert^2} 
 \left(C_{(0, \frac{\gamma}{2},0)} m^\frac{\gamma}{2} |v-v_*|^\gamma \right. \\
  \left. +C_{(\frac{\gamma}{2},0, \frac{\gamma}{2})}\left(  I^\frac{\gamma}{2}+(k T)^{  \frac{\gamma}{2}} \frac{\Gamma\!\left(\alpha+\frac{\gamma}{2}+1\right)}{\Gamma(\alpha+1) } \right) \right)  \mathrm{d} v_*.
\end{multline}
It remains to perform the integration with respect to the velocity $v_*$, and there appear  two integrals to calculate
\begin{equation*}
\mathcal I_1 = \int_{\mathbb R^3} e^{-\frac{m}{2 k T}  \vert v_*-U\vert^2} \mathrm d v_*, \quad \mathcal I_2 =\int_{\mathbb R^3} e^{-\frac{m}{2 k T}  \vert v_*-U\vert^2} \vert v-v_*\vert ^\gamma \mathrm{d} v_*.
\end{equation*}
For the first integral $\mathcal I_1$ we immediately get the result, while for the second one $\mathcal I_2$ we   first change the velocity $v_*$ into the relative velocity $u=v-v_*$ and use the peculiar velocity $c=v-U$ to obtain
\begin{equation*}
\mathcal I_2=\int_{\mathbb{R}^3}  e^{-\frac{m}{2 k T}  \vert v_*-U\vert^2}  |v-v_*|^\gamma \mathrm d v_* =\int_{\mathbb{R}^3}  e^{-\frac{m}{2 k T}  \vert u-c\vert^2}  |u|^\gamma \mathrm{d}u .
\end{equation*}
Then the relative velocity $u$ is expressed in terms of the spherical coordinates with the zenith $c$, and the angle $\varphi$ between $\frac{u}{\left| u \right|}$  and $\frac{c}{\left| c \right|}$  so that $\cos \varphi = \frac{u}{\left| u \right|} \cdot  \frac{c}{\left| c \right|} $,
\begin{equation*}
\mathcal I_2 = 2 \pi \int_{[0,\infty) \times [0, \pi) } e^{-\frac{m}{2 k T}  \left(\vert u\vert^2 + \vert c\vert^2 - 2 \vert u\vert \vert c\vert \cos \varphi  \right)} \vert u\vert^{\gamma+2} \sin \varphi \; \mathrm{d} \vert u\vert \,\mathrm{d} \varphi.
\end{equation*}
Knowing that
\begin{equation*}
\int_{0}^\pi e^{\frac{m}{ k T}   \vert u\vert \vert c\vert \cos \varphi } \sin \varphi  \,\mathrm{d} \varphi = \frac{2 k T}{m   \vert u\vert \vert c\vert} \sinh\left( \frac{m}{k T}   \vert u\vert \vert c\vert  \right),
\end{equation*}
integration with respect to $\varphi$ for $\mathcal I_2$ yields
\begin{equation*}
\mathcal I_2=  4 \pi \frac{ k T}{m \vert c \vert}    e^{-\frac{m}{2 k T} \vert c\vert^2   }  \int_{[0, \infty)} e^{-\frac{m}{2 k T}   \vert u\vert^2  } \sinh \left(\frac{m \vert c \vert \vert u \vert}{k T}\right) \vert u \vert^{\gamma+1} \mathrm d \vert u \vert.
\end{equation*}
Next we write the special value of the confluent hypergeometric function  $\, _1F_1(a,b,z)$ when evaluated at $b=\frac{3}{2}$, see  Chapter 13 in \cite{Abr},  
\begin{equation}\label{hyp geo}
	_1F_1\left(a,\tfrac{3}{2},z\right)=\frac{1}{\Gamma(a) }\int_{0}^{\infty} e^{- t}t^{a-1} \frac{\sinh(2 \sqrt{z t})}{2\sqrt{z t}} \mathrm{d}t,\quad  a> 0.
\end{equation}
Using this representation,  the integral $\mathcal{I}_2$ can be rewritten as
\begin{equation*}
\mathcal I_2= 2\pi \left(\frac{2 k T}{m}\right)^{\frac{\gamma+3}{2} } \Gamma\! \left(\frac{\gamma+3}{2}  \right) \, e^{-\frac{m}{2 k T} \vert c\vert^2   } \ _1F_1\!\left(\frac{\gamma+3}{2}, \frac{3}{2}, \frac{m \vert c \vert ^2}{2 k T}\right).
\end{equation*}
Finally, the collision frequency $\nu_{\gamma,\alpha}(v, I)$ from \eqref{pomocna coll}  becomes

\begin{multline}\label{pomocna coll 2}
	\nu_{\gamma , \alpha}(v, I)=  \frac{\rho}{m}  \frac{ \left\| b \right\|_{L^1(\mathrm{d}\sigma)}}{m^\frac{\gamma}{2}} \\ \times
 \left(C_{(0, \frac{\gamma}{2},0)} \frac{2^{\frac{\gamma}{2}+1}}{\sqrt{\pi}}  \left( k T  \right)^{\frac{\gamma}{2}}     \Gamma\! \left(\frac{\gamma+3}{2}  \right) \, e^{-\frac{m}{2 k T} \vert c\vert^2   } \ _1F_1 \left(\frac{\gamma+3}{2}, \frac{3}{2}, \frac{m \vert c \vert ^2}{2 k T}\right) 
 \right. \\
	\left. +C_{(\frac{\gamma}{2},0, \frac{\gamma}{2})}\left(  I^\frac{\gamma}{2}+(k T)^{  \frac{\gamma}{2}} \frac{\Gamma\!\left(\alpha+\frac{\gamma}{2}+1\right)}{\Gamma(\alpha+1) } \right) \right) 
	\\
	=  \frac{\rho}{m}  \left\| b \right\|_{L^1(\mathrm{d}\sigma)}  \left( \frac{p}{\rho}  \right)^{\frac{\gamma}{2}}  \frac{\Gamma (\alpha +1)}{\Gamma\!\left(\frac{4 \alpha +\gamma +7}{2}\right)}      \\ \times
	\left( \Gamma (\alpha +1) \Gamma\! \left(\frac{\gamma +3}{2}\right)^2 \frac{2^{\frac{\gamma}{2}+1}}{\sqrt{\pi}}     \, e^{-\frac{m}{2 k T} \vert c\vert^2   } \ _1F_1\! \left(\frac{\gamma+3}{2}, \frac{3}{2}, \frac{m \vert c \vert ^2}{2 k T}\right) 
	\right. \\
	\left. +  \frac{\sqrt{\pi } }{2} \Gamma \!\left(\alpha +\frac{\gamma }{2}+1\right) \left(  \left(\frac{I}{k T}\right)^\frac{\gamma}{2}+  \frac{\Gamma\! \left(\alpha+\frac{\gamma}{2}+1\right)}{\Gamma(\alpha+1) } \right) \right),   
\end{multline}
where we have used the relation \eqref{hydro p}. Then \eqref{pomocna coll 2}   is precisely \eqref{coll freq1}.

\section{Computation of the production terms  for the six fields model}\label{App production 6 mom entr}

For the sake of simplicity,  we introduce the  following notation 
\begin{equation}\label{M N L}
M= \frac{m}{2\,kT}  \frac{1}{\left(1+\frac{\Pi}{p}\right) }, \  N = \frac{1}{kT} \left( \frac{1}{1-\frac{3}{2(\alpha+1)} \frac{\Pi}{p}} \right), \ L= \frac{\rho}{m} \left( \frac{M}{\pi}\right)^{\frac{3}{2}} \frac{N^{\alpha+1}}{\Gamma(\alpha+1)}.
\end{equation}
We remind that the range for $\Pi/p$ imposed in \eqref{range Pi/p} implies positivity of these coefficients,
\begin{equation*}\label{M N pos}
M > 0, \quad N>0.
\end{equation*}
With these coefficients, six moments distribution function \eqref{f six} reads
\begin{equation}\label{f 6 LMN}
\hat{f}^6 =L \,  I^\alpha \,  e^{- M \left|c\right|^2 - N I}.
\end{equation}
The aim here is to calculate the production term $\mathcal{P}$ defined in \eqref{productionPii:equ} and the entropy production term $D^{nw}(\hat{f}^6)$ from \eqref{entr prod nw} for the  cross section \eqref{model 3}, namely
\begin{multline*}
\mathcal{B}^{nw}(v, v_*, I, I_*, r, R, \sigma)  \\
= b\big(\tfrac{u}{\left| u\right|} \cdot \sigma\big)  \left( R^{\frac{\gamma}{2}} |u|^\gamma  + \left( r (1-R)\frac{I}{m} \right)^{\frac{\gamma}{2}} + \left( (1-r) (1-R)\frac{I_*}{m} \right)^{\frac{\gamma}{2}} \right),   \quad \gamma >0
\end{multline*}
where we have denoted $u:=v-v_*$.

\subsection{Computation of the production term $\mathcal{P}$}\label{App production 6}

For the production term  \eqref{productionPii:equ}  we first note that taking the square of $v=c+U$, it reduces to
\begin{equation*}
\mathcal{P} = \int_{\mathbb{R}^3 \times [0,\infty)} m \left| c \right|^2 Q^{nw}(\hat{f}_6, \hat{f}_6)(c,I) \, \mathrm{d} I  \, \mathrm{d} c.
\end{equation*}
where we remind that the collision operator $ Q^{nw}$ is defined in \eqref{Q non-weight pull out}. The weak form  \eqref{weak 1} yields
\begin{multline}\label{P medju}
	\mathcal{P}=\frac{m}{2}\int_{\mathbb{R}^6 \times [0,\infty)^2 \times [0,1]^2 \times S^2 } \left( \vert c' \vert^2 + \vert c'_* \vert^2 -\vert c \vert^2 - \vert c_*\vert^2\right) \hat{f}_6\, \hat{f}_{6*} \\ \phantom{\mathbb{R}^6 \times}
	\times\mathcal{B}^{nw} \phi_\alpha(r)  \,  (1-R) R^{\frac{1}{2}} \psi_\alpha(R)\,   \mathrm{d} \sigma \, \mathrm{d} r \, \mathrm{d} R \,  \mathrm{d} I_* \, \mathrm{d}  c_*\, \mathrm{d} I \, \mathrm{d}  c \\
	=\frac{m}{2} L^2 \int_{\mathbb{R}^6 \times [0,\infty)^2 \times [0,1]^2 \times S^2 }\left( \vert c' \vert^2 + \vert c'_* \vert^2 -\vert c \vert^2 - \vert c_*\vert^2\right)e^{- M \left( \vert c \vert^2 +\vert c_* \vert^2\right)} e^{ - N\, \left(I+I_*\right)} \\
	\\ \phantom{\mathbb{R}^6 \times}
	\times \mathcal{B}^{nw} \phi_\alpha(r)  \,  (1-R) R^{\frac{1}{2}} \psi_\alpha(R)\, I^{\alpha}  I_*^{\alpha} \,  \mathrm{d} \sigma \, \mathrm{d} r \, \mathrm{d} R \,  \mathrm{d} I_* \, \mathrm{d}  c_*\, \mathrm{d} I \, \mathrm{d}  c, 
\end{multline}
where $\phi_\alpha(r)$ and $\psi_\alpha(R)$ are defined in \eqref{fun r R} and $\mathcal{B}^{nw}$ is cross section \eqref{model 3}. Now we pass to the relative velocity $u$ and center of mass peculiar  velocity $V_c$ by means of the following change of variables 
\begin{equation}\label{c notation}
\left(c, c_*\right) \mapsto \left(u:= c-c_*, V_c = \frac{c + c_*}{2}\right) \quad \Rightarrow \quad	c=V_c + \frac{u}2, \ c_*=V_c-\frac{u}2,
\end{equation}
with unit Jacobian. Therefore, the terms under integral in new variables become
\begin{equation*}
\vert c' \vert^2 + \vert c'_* \vert^2 -\vert c \vert^2 - \vert c_*\vert^2=\frac{1}{2}\left(R-1\right)\vert  u \vert^2 + \frac{2R}{m}\left( I+ I_*\right)
\end{equation*}
and
\begin{equation}\label{ e podint c^2+c*^2}
	\vert c \vert^2 +\vert c_* \vert^2= 2 \vert V_c\vert^2 + \frac{1}{2}\vert u \vert^2.
\end{equation}
Therefore, we can express the primed quantities from \eqref{P medju} in center-of-mass framework,
\begin{multline}\label{P medju 2}
\mathcal{P}
=\frac{m}{2} L^2 \int_{\mathbb{R}^6 \times [0,\infty)^2 \times [0,1]^2 \times S^2 }
\left( \frac{1}{2}\left(R-1\right)\vert  u \vert^2 + \frac{2R}{m}\left( I+ I_*\right) \right)e^{- M \left( 2 \vert V_c\vert^2 + \frac{1}{2}\vert u \vert^2\right)}  \\
\\ 
\times e^{ - N\, \left(I+I_*\right)} \mathcal{B}^{nw} \phi_\alpha(r)  \,  (1-R) R^{\frac{1}{2}} \psi_\alpha(R)\, I^{\alpha}  I_*^{\alpha} \,  \mathrm{d} \sigma \, \mathrm{d} r \, \mathrm{d} R \,  \mathrm{d} I_* \, \mathrm{d}  c_*\, \mathrm{d} I \, \mathrm{d}  c, 
\end{multline}
Using that the cross section is of the form \eqref{model 3}, i.e. 
\begin{equation*}
\mathcal{B}^{nw}(v, v_*, I, I_*, r, R, \sigma)  \\
= b\big(\tfrac{u}{\left| u\right|}\cdot\sigma\big) \,  \tilde{B}(\left| u \right|, I, I_*, r, R),
\end{equation*}
we can perform the integration with respect to $V_c$ and $\sigma$,
\begin{multline*}
	\mathcal{P}=
	m\,L^2\left(\frac{ \pi }{2\, M}\right)^{\frac{3}{2}} \left\| b \right\|_{L^1(\mathrm{d}\sigma)} \int_{\mathbb{R}^3 \times [0,\infty)^2 \times [0,1]^2  } 
	e^{- \frac{M}{2}  \vert u \vert^2 } e^{ - N\, \left(I+I_*\right)} \\
	\times \left(\frac{1}{4}\left(R-1\right)\vert  u \vert^2 + \frac{R}{m}\left( I+ I_*\right)\right) \tilde{B}(\left| u \right|, I, I_*, r, R)
	\\ 
	\times \phi_\alpha(r)  \,  (1-R) R^{\frac{1}{2}} \psi_\alpha(R)\, I^{\alpha}  I_*^{\alpha} \,  \mathrm{d} r \, \mathrm{d} R \,  \mathrm{d} I_* \, \mathrm{d} I \, \mathrm{d}  u. 
\end{multline*}
We now pass to the spherical coordinates for the relative velocity $u$. Denoting $y=\left|u\right|$, and performing integration with respect to the angular part, we obtain
\begin{multline*}
\mathcal{P}= m\,L^2\left(\frac{ \pi }{2\, M}\right)^{\frac{3}{2}}\left\| b \right\|_{L^1(\mathrm{d} \sigma)} \, 4 \,\pi \int_{ [0,\infty)^3 \times [0,1]^2 } 
	e^{- \frac{M}{2} y^2} e^{ - N\, \left(I+I_*\right)} y^2  \\
\times \left(\frac{1}{4}(R-1)\,y^2 + \frac{R}{m}\left( I+ I_*\right)\right)\\
\times	 \left( R^{\gamma/2} y^\gamma  + \left( r (1-R)\frac{I}{m} \right)^{\gamma/2} + \left( (1-r) (1-R)\frac{I_*}{m} \right)^{\gamma/2} \right), \\
\times\phi_\alpha(r)  \,  (1-R) R^{\frac{1}{2}} \psi_\alpha(R)\, I^{\alpha}  I_*^{\alpha}  \, \mathrm{d} r \, \mathrm{d} R \,  \mathrm{d} I_* \, \mathrm{d} I \, \mathrm{d}  y. 
\end{multline*}
We expand all the expressions involved and perform integration with respect to $r$ and $R$. Using the notation \eqref{const r R} for the constant issuing from this integration,  the production term becomes
\begin{multline*}
\mathcal{P} =	m\,L^2\left(\frac{ \pi }{2\, M}\right)^{\frac{3}{2}}\left\| b \right\|_{L^1(\mathrm{d} \sigma)} \, 4 \,\pi \int_{ [0,\infty)^3 \times [0,1]^2 } 	e^{- \frac{M}{2} y^2} e^{ - N\, \left(I+I_*\right)} \,  I^{\alpha}  I_*^{\alpha}
	\\ \times
	\left\{    -\frac{1}{4} C_{\left(1,\frac{\gamma}{2},0\right)} \,y^{ \gamma+2}  +  C_{\left(0,\frac{\gamma}{2}+1,0\right)} \, y^{\gamma} \left( \frac{I}{m}+ \frac{I_*}{m}\right)
		\right.\\\left.
  - \frac{1}{4} C_{\left(\frac{\gamma}{2}+1,0,\frac{\gamma}{2}\right)}  \, y^2 \left( \left(\frac{I}{m}\right)^{\frac{\gamma}{2}} + \left(\frac{I_*}{m}\right)^{\frac{\gamma}{2}}  \right)
	\right.\\\left.
+  C_{\left(\frac{\gamma}{2},1,\frac{\gamma}{2}\right)} \left(   \left(\frac{I}{m}\right)^{\frac{\gamma}{2}+1} + \left(\frac{I}{m}\right)^{\frac{\gamma}{2}}\frac{I_*}{m} + \left(\frac{I_*}{m}\right)^{\frac{\gamma}{2}}\frac{I}{m}  + \left(\frac{I_*}{m}\right)^{\frac{\gamma}{2}+1} \right)\right\} 
  \mathrm{d} I_* \, \mathrm{d} I \, \mathrm{d}  y.
	\end{multline*}
It remains to integrate with respect to  $I,\, I_*$ and $y$. Introducing the positive function
\begin{multline}\label{Cp LMN}
\tilde{\mathcal{C}}_{\mathcal{P}} =	m\, L^2\left(\frac{ \pi }{2\, M}\right)^{\frac{3}{2}}\,\left\| b \right\|_{L^1(\mathrm{d} \sigma)} \, 4 \,\pi 
\,\frac{\Gamma(\alpha +1)^2}{\Gamma\!\left(\frac{4\alpha+ \gamma+9}{2}\right) }\,
N^{-(2\alpha+2)}\, M^{-\frac{3}{2}}\\
\times
 \left( k_1 \,M^{-\frac{\gamma}{2}}+k_2\,(m\, N )^{-\frac{\gamma}{2}} \right) \geq 0,
\end{multline}
where $k_1$ and $k_2$ are from \eqref{k1-k2},
we finally obtain the expression for the production term
	\begin{equation}\label{ProdEnt P preko Cp}
	\mathcal{P}= \tilde{\mathcal{C}}_{\mathcal{P}}  \left(-\frac{1}{2}M^{-1}+(m\, N)^{-1}\right).
	\end{equation}
It remains to come back to the original variables using \eqref{M N L}, that yields
\begin{equation}\label{pommm}
	 \left(-\frac{1}{2}M^{-1}+(m\, N)^{-1}\right) = - \frac{k T}{m} \frac{(\alpha+\frac{5}{2})}{(\alpha+1)} \frac{\Pi}{p} = - \frac{p}{\rho} \frac{(\alpha+\frac{5}{2})}{(\alpha+1)} \frac{\Pi}{p}.
\end{equation}
Gathering \eqref{Cp LMN}, \eqref{ProdEnt P preko Cp} and \eqref{pommm} yields  \eqref{P iz Prop}.

\subsection{Computation of the entropy production term $D^{nw}(\hat{f}^6)$}\label{App production 6 entr}
With the notation \eqref{f 6 LMN}, the test function corresponding to the entropy law becomes
\begin{equation*}
\log\left(\hat{f}^6 I^{-\alpha} \right) = \log  L - M \left|c\right|^2 - N I.
 \end{equation*}
 The weak form \eqref{weak form nw} allows to write
 \begin{multline}\label{D nw medju}
D^{nw}(\hat{f}^6) 
 = \frac{1}{2}L^2 \int_{\mathbb{R}^6 \times [0,\infty)^2 \times [0,1]^2 \times S^2 } \left( -M \left( \vert c' \vert^2 +\vert c'_* \vert^2-\vert c \vert^2-\vert c_* \vert^2\right) \right.\\\left.- N \left(I'+I'_*-I-I_*\right) \right)
 \times e^{- M \left( \vert c \vert^2 +\vert c_* \vert^2\right)} e^{ - N\, \left(I+I_*\right)} \\
  \mathcal{B}^{nw} \phi_\alpha(r)  \,  (1-R) R^{\frac{1}{2}} \psi_\alpha(R)\, I^{\alpha}  I_*^{\alpha} \,  \mathrm{d} \sigma \, \mathrm{d} r \, \mathrm{d} R \,  \mathrm{d} I_* \, \mathrm{d}  c_*\, \mathrm{d} I \, \mathrm{d}  c, 
 \end{multline}
 with $\phi_\alpha(r)$ and $\psi_\alpha(R)$ from \eqref{fun r R} and $\mathcal{B}^{nw}$ is the cross section \eqref{model 3}.  The next step is to use coordinates of the center of mass by means of \eqref{c notation}. Indeed, in addition to \eqref{micro CL cm} and \eqref{ e podint c^2+c*^2} we also have
 \begin{equation*}
 I'+I'_*-I-I_*=-\frac{m}{2}\left(\frac{1}{2}(R-1)\,\vert u \vert^2 + \frac{2 R}{m}\left( I+ I_*\right)\right).
 \end{equation*}
 These considerations allow to write \eqref{D nw medju}  in terms of the production term $\mathcal{P}$ by virtue of \eqref{P medju 2},
\begin{equation}\label{D LMN}
D^{nw}(\hat{f}^6)=\left(\frac{m}{2}N-M\right)\frac{\mathcal{P}}{m}.
\end{equation}
Therefore, using the results of the previous Section \ref{App production 6} and notably its final result \eqref{ProdEnt P preko Cp} we obtain
\begin{equation*}
D^{nw}(\hat{f}^6)=-\left(\frac{m}{2}N-M\right)^2\frac{\tilde{\mathcal{C}}_\mathcal{P}}{m \, N \, M},
\end{equation*}
with the positive constant $\mathcal{C}_\mathcal{P}$ from \eqref{Cp LMN}. Now is clear that  $D^{nw}(f)$ is non-positive, as claimed in the H-theorem \eqref{H th non positive}.\\

The final result follows from \eqref{D LMN}  by exploiting
\begin{equation*}
\left(\frac{m}{2}N-M\right) = \frac{m}{2 k T}  \left(  1- \frac{3}{2 (\alpha + 1)} \frac{\Pi}{p} \right)^{-1}\left(  1 + \frac{\Pi}{p} \right)^{-1} \frac{(\alpha + \frac{5}{2})}{(\alpha+1)} \frac{\Pi}{p}.
\end{equation*}

We note that the shorter notation in terms of $M,N$, allows to rewrite the derivative \eqref{k der},
\begin{equation*}
	\mathcal {K}_\Pi= - \frac{3\,k }{m}{\left( \frac{m}{2}N -M\right)}.
\end{equation*}
Combining the last equation with \eqref{ProdEnt P preko Cp} we get
\begin{equation*}
\frac{1}{3} \frac{\partial \mathcal {K}}{\partial \Pi} \mathcal{P} =	-\frac{k }{m}{\left( \frac{m}{2}N -M\right)}\mathcal{P}=-k D^{nw}(\hat{f}^6) \geq 0.
\end{equation*}

\section{Computation of the production terms  for the fourteen fields model}\label{App production14 mom}
For the sake of simplicity, the equilibrium distribution function \eqref{Maks} will be written as in \eqref{Maxwell dmsless},
\begin{equation*}\label{Maxwell-simple}
	f_M=I^{\alpha} L_0  e^{-\frac{1}{k T} (\frac{m}{2} \vert c\vert^2 +I)} \quad \text{where} \quad L_0= \frac{\rho}{m} \left( \frac{m}{2 \pi k T} \right)^{\frac{3}{2}} \frac{1}{\Gamma(\alpha + 1)} \frac{1}{\left(k T\right)^{\alpha+1}}.
\end{equation*}

Our aim is to compute the production terms in the fourteen moments approximation, which amounts to plug the approximative distribution function $f^{14}$ into the definition of the product terms \eqref{14 prod},
\begin{equation*}
P_{ij}^{14}= \int_{\mathbb{R}^3 \times \mathbb{R}_+} m v_i v_j Q^{nw}(\hat{f}^{14},\hat{f}^{14})(v,I)\mathrm{d} I \, \mathrm{d}{v},
\end{equation*}
\begin{equation*}
Q_{i}^{14}= \int_{\mathbb{R}^3 \times \mathbb{R}_+}\left(\frac{m}{2} \vert v \vert ^2+I\right) v_i\, Q^{nw}(\hat{f}^{14},\hat{f}^{14})(v,I) \, \mathrm{d} I \, \mathrm{d}{v}.
\end{equation*}
  Introducing the peculiar velocity $c=v-U$ and using annihilations of the collision operator weak form \eqref{weak form coll inv}, and after the change of variables $v\mapsto c$ the expressions  \eqref{14 prod}  simplify to
\begin{equation}\label{Prod14}
P_{ij}^{14}= \int_{\mathbb{R}^3 \times \mathbb{R}_+} m c_i c_j Q^{nw}(\hat{f}^{14},\hat{f}^{14})(c+U,I)\mathrm{d} I \, \mathrm{d}{c}, 
\end{equation}
\begin{equation}\label{Qflux14}
Q_{i}^{14}= \sum_{k=1}^3 U_k P_{ki}+ \int_{\mathbb{R}^3 \times \mathbb{R}_+} c_i\left(\frac{m}{2} \vert c \vert ^2+I\right) \, Q^{nw}(\hat{f}^{14},\hat{f}^{14})(c+U,I)  \mathrm{d} I \, \mathrm{d}{c}.
\end{equation}
As non-equilibrium effects are supposed to be small, products of the distribution functions appearing in the collision integral can be linearized with respect to the non-equilibrium quantities, $p_{\langle ij \rangle}, \Pi, q_i$. Using the microscopic conservation laws \eqref{micro CL}, it follows
\begin{multline}\label{LinQI}
\hat{f}^{14'}\hat{f}_{*}^{14'}-\hat{f}^{14}\hat{f}_{*}^{14}\approx f_M f_{M*} \left\{ \sum_{k,l=1} ^3 \frac{\rho}{2p^2} \left(p_{\langle kl \rangle}+\left(\alpha+\frac{5}{2}\right)(\alpha+1)^{-1} \Pi \delta_{kl} \right)\right.\\
\left. \times (c'_{k}c'_{l}+c'_{*k}c'_{*l}-c_{k}c_{l}-c_{*k}c_{*l}) +\sum_{n=1}^{3} \left(\alpha+\frac{7}{2}\right)^{-1} \frac{\rho ^2}{m p^3} q_n \right. \\
\left.\times \left( \left(\frac{m}{2} \vert c' \vert ^2+I'\right)c'_{n}+\left(\frac{m}{2} \vert c'_{*} \vert ^2+I'_{*}\right)c'_{*n}-\left(\frac{m}{2} \vert c \vert ^2+I\right)c_{n}-\left(\frac{m}{2} \vert c_{*} \vert ^2+I_{*}\right)c_{*n}\right)\right\}.
\end{multline}
Placing \eqref{LinQI} into \eqref{Prod14} and \eqref{Qflux14} yields  a suitable approximation for the source terms  $P_{ij}^{14}$ and $Q_i ^{14}$, denoted by $\overline{P}_{ij}^{14}$ and $\overline{Q}_i ^{14}$, respectively. 

We now introduce the following notation,
\begin{multline*}
\mathcal{P}_{ijkl}=  \int m\, c_i \, c_j  \left(c'_{k}c'_{l}+c'_{*k}c'_{*l}-c_{k}c_{l}-c_{*k}c_{*l}\right) f_M f_{M*}\\
\times  \mathcal{B}^{nw} \phi_\alpha(r)  \,  (1-R) R^{\frac{1}{2}} \psi_\alpha(R)\, \mathrm{d} \sigma \, \mathrm{d} r \, \mathrm{d} R \,  \mathrm{d} I_* \, \mathrm{d}  c_*\, \mathrm{d} I \, \mathrm{d}  c ,
\end{multline*}
\begin{multline*}
\mathcal{Q}_{in}= \int \left(\frac{m}{2} \vert c \vert ^2+I\right)c_{i}   \left( \left(\frac{m}{2} \vert c' \vert ^2+I'\right)c'_{n}+\left(\frac{m}{2} \vert c'_{*} \vert ^2+I'_{*}\right)c'_{*n}\right.\\
\left. -\left(\frac{m}{2} \vert c \vert ^2+I\right)c_{n}-\left(\frac{m}{2} \vert c_{*} \vert ^2+I_{*}\right)c_{*n}\right) f_M f_{M*} \\
\times \mathcal{B}^{nw} \phi_\alpha(r)  \,  (1-R) R^{\frac{1}{2}} \psi_\alpha(R)\,  \mathrm{d} \sigma \, \mathrm{d} r \, \mathrm{d} R \,  \mathrm{d} I_* \, \mathrm{d}  c_*\, \mathrm{d} I \, \mathrm{d}  c.
\end{multline*}
Now the parity arguments imply 
\begin{equation}\label{Pbar}
\overline P_{ij}^{14}=\frac{\rho}{2p^2}  \sum_{k,l=1} ^3  \left( p_{\langle kl\rangle}+\left(\alpha+\frac{5}{2}\right)(\alpha+1)^{-1} \Pi \delta_{kl} \right) \mathcal{P}_{ijkl}, 
\end{equation}
\begin{equation}\label{Qbar}
\overline Q_{i}^{14}=\sum_{k=1}^{3} U_k \overline P_{ki}^{14}+\sum_{n=1}^{3} \left(\alpha+\frac{7}{2}\right)^{-1} \frac{\rho ^2}{m p^3} q_n \mathcal{Q}_{in}. 
\end{equation}
We calculate the production  terms $\overline P_{ij}^{14}$ and $\overline Q_{i}^{14}$ in separate sections for the cross section
\begin{multline}\label{model 3 const} 
\mathcal{B}^{nw}(v, v_*, I, I_*, r, R, \sigma)  \\
= K \left( R^{\frac{\gamma}{2}} |u|^\gamma  + \left( r (1-R)\frac{I}{m} \right)^{\frac{\gamma}{2}} + \left( (1-r) (1-R)\frac{I_*}{m} \right)^{\frac{\gamma}{2}} \right),  
\end{multline}
where $K$ is a constant, $u:=v-v_*$, $\gamma>0$.

\subsection{Computation of $\overline P_{ij}^{14}$} 
Firstly, we exploit the parity arguments for the term $\mathcal{P}_{ijkl}$. Note that it vanishes unless indices are equal by pairs - the integral is non zero when $i=j$ and $k=l$ or $i=k$ and $j=l$ or $i=l$ and $j=k$. By symmetry, the last two terms lead to the same result and thus $\mathcal{P}_{ijkl}$ can be represented in the following form:
\begin{equation}
\mathcal{P}_{ijkl}=\mathcal{P}_1 \delta_{ij} \delta_{kl}+\mathcal{P}_2 (\delta_{ik}\delta_{jl}+\delta_{il} \delta{jk}). \label{Pijkl}
\end{equation}
Combining \eqref{Pbar} and \eqref{Pijkl} with the fact that pressure tensor is symmetric, we get
\begin{equation*}
\overline P_{ij}^{14}= \frac{\rho}{2p^2}\left(2 p_{\langle ij\rangle}\mathcal P_2+\frac{1}{3} \delta_{ij} \left(\alpha+\frac{5}{2}\right)(\alpha+1)^{-1} \Pi\sum_{r,t=1}^3 \mathcal P_{rrtt} \right).
\end{equation*}
The term $\mathcal P_2$ can be determined from the system of equations obtained from the representation  \eqref{Pijkl}
\[\sum_{r,t=1}^{3} \mathcal P_{rrtt}=9 \mathcal P_1+6\mathcal P_2, \quad \quad \sum_{r,t=1}^3 \mathcal P_{rtrt}=3\mathcal P_1+12 \mathcal P_2,
\]
whose solution is
\[\mathcal P_1=\frac{1}{15} \sum_{r,t=1} ^3 (2 \mathcal P_{rrtt}-\mathcal P_{rtrt}), \quad \quad \mathcal P_2=\frac{1}{30} \sum_{r,t=1} ^3 (3 \mathcal P_{rtrt}-\mathcal P_{rrtt}).
\]
		
\subsubsection{Computation of $\sum_{r,t=1}^3 \mathcal P_{rrtt}$}
We first concentrate on the term 
\begin{multline*}
 \sum_{r,t=1}^3 \mathcal P_{rrtt}={m}  L_0^2 \int_{\mathbb{R}^6 \times [0,\infty)^2 \times [0,1]^2 \times S^2}  e^{-\frac{1}{k T}\left( \frac{m}{2} (\vert c \vert ^2 + \vert c_* \vert^2)+I+I_*\right)} \\
 \times \vert c \vert^2 \left( \vert c' \vert^2 +\vert c'_* \vert^2-\vert c \vert^2-\vert c_* \vert^2\right)\\
 \times \mathcal{B}^{nw} \phi_\alpha(r)  \,  (1-R) R^{\frac{1}{2}} \psi_\alpha(R)\, I^{\alpha}  I_*^{\alpha} \,  \mathrm{d} \sigma \, \mathrm{d} r \, \mathrm{d} R \,  \mathrm{d} I_* \, \mathrm{d}  c_*\, \mathrm{d} I \, \mathrm{d}  c.
\end{multline*}
Passing to the center-of-mass reference frame, by changing  variables to \eqref{c notation} yields
\begin{multline*}
 \sum_{r,t=1}^3 \mathcal P_{rrtt}=m  L_0^2 \int_{\mathbb{R}^6 \times [0,\infty)^2 \times [0,1]^2 \times S^2}  e^{-\frac{1}{k T}\left( m \vert V_c\vert^2 +\frac{m}{4} \vert u \vert ^2+I+I_*\right)} \\
 \times \left(\vert V_c \vert^2+ V_c \cdot u+\frac{1}{4} \vert u \vert^2 \right) \left( \frac{1}{2}\left(R-1\right)\vert  u \vert^2 + \frac{2R}{m}\left( I+ I_*\right) \right)\\
 \times \mathcal{B}^{nw} \phi_\alpha(r)  \,  (1-R) R^{\frac{1}{2}} \psi_\alpha(R)\, I^{\alpha}  I_*^{\alpha} \,  \mathrm{d} \sigma \, \mathrm{d} r \, \mathrm{d} R \,  \mathrm{d} I_* \, \mathrm{d} I \, \mathrm{d}  u \, \mathrm{d}  V_c.
\end{multline*}
The form of the cross-section \eqref{model 3 const}
\begin{equation}\label{model Bnw u}
\mathcal{B}^{nw}(v, v_*, I, I_*, r, R, \sigma)  \\
= K\tilde{B}(\left| u \right|, I, I_*, r, R),
\end{equation}
allows to immediately integrate  with respect to $V_c$ and $\sigma$,
\begin{multline*}
 \sum_{r,t=1}^3 \mathcal P_{rrtt}=m K L_0^2  \left( \frac{\pi k T}{m} \right)^{\frac{3}{2}} 2\pi  \int_{\mathbb{R}^3 \times [0,\infty)^2 \times [0,1]^2} e^{-\frac{1}{k T}\left(\frac{m}{4} \vert u \vert ^2+I+I_*\right)} \\
 \times \left(3 \frac{k T}{m}+\frac{1}{2} \vert u \vert ^2\right)\left( \frac{1}{2}\left(R-1\right)\vert  u \vert^2 + \frac{2R}{m}\left( I+ I_*\right) \right)\\
 \times \tilde{B} \phi_\alpha(r)  \,  (1-R) R^{\frac{1}{2}} \psi_\alpha(R)\, I^{\alpha}  I_*^{\alpha}  \, \mathrm{d} r \, \mathrm{d} R \,  \mathrm{d} I_* \, \mathrm{d} I \, \mathrm{d}  u.
\end{multline*}
Next, we pass to the spherical coordinates for the relative velocity $u$. Denoting $y=\vert u \vert$, and performing integration with respect to the angular part, we obtain
\begin{multline*}
 \sum_{r,t=1}^3 \mathcal P_{rrtt}=m K L_0^2  \left( \frac{\pi k T}{m} \right)^{\frac{3}{2}} 8 \pi ^2  \int_{[0,\infty)^3 \times [0,1]^2} e^{-\frac{1}{k T}\left(\frac{m}{4} y ^2+I+I_*\right)} \\
 \times \left(3 \frac{k T}{m}+\frac{1}{2} y ^2\right)\left( \frac{1}{2}\left(R-1\right)y+ \frac{2R}{m}\left( I+ I_*\right) \right)\\
 \times  \left( R^{\frac{\gamma}{2}} |u|^\gamma  + \left( r (1-R)\frac{I}{m} \right)^{\frac{\gamma}{2}} + \left( (1-r) (1-R)\frac{I_*}{m} \right)^{\frac{\gamma}{2}} \right)\\
 \times y^2 \phi_\alpha(r)  \,  (1-R) R^{\frac{1}{2}} \psi_\alpha(R)\, I^{\alpha}  I_*^{\alpha}  \, \mathrm{d} r \, \mathrm{d} R \,  \mathrm{d} I_* \, \mathrm{d} I \, \mathrm{d}  y.
\end{multline*}
Now we expand all the involved expressions  and perform integration with respect to $r$ and $R$.   Using the notation \eqref{const r R}, we obtain
\begin{multline*}
 \sum_{r,t=1}^3 \mathcal P_{rrtt}=m K L_0^2  \left( \frac{\pi k T}{m} \right)^{\frac{3}{2}} 8 \pi ^2  \int_{[0,\infty)^3} e^{-\frac{1}{k T}\left(\frac{m}{4} y ^2+I+I_*\right)}y^2 I^{\alpha}  I_*^{\alpha}  \\
 \times \left(3 \frac{k T}{m}+\frac{1}{2} y ^2\right) \left\{ - \frac{1}{2} C_{\left(1,\frac{\gamma}{2},0\right)} \,y^{ \gamma+2} +  2C_{\left(0,\frac{\gamma}{2}+1,0\right)} \, y^{\gamma} \left( \frac{I}{m}+ \frac{I_*}{m}\right)
 	\right.\\
 \left. -\frac{1}{2}  C_{\left(\frac{\gamma}{2}+1,0,\frac{\gamma}{2}\right)} \, y^2 \left( \left(\frac{I}{m}\right)^{\frac{\gamma}{2}} + \left(\frac{I_*}{m}\right)^{\frac{\gamma}{2}}  \right)
	\right.\\
	\left.+ C_{\left(\frac{\gamma}{2},0,\frac{\gamma}{2}\right)} \left(   \left(\frac{I}{m}\right)^{\frac{\gamma}{2}+1} + \left(\frac{I}{m}\right)^{\frac{\gamma}{2}}\frac{I_*}{m} + \left(\frac{I_*}{m}\right)^{\frac{\gamma}{2}}\frac{I}{m}  + \left(\frac{I_*}{m}\right)^{\frac{\gamma}{2}+1} \right)\right\} 
  \mathrm{d} I_* \, \mathrm{d} I \, \mathrm{d}  y.  
\end{multline*}
Finally, after performing integration with respect to $I$, $I_*$ and $y$, we get
\begin{multline*}
\sum_{r,t=1}^3\mathcal{P}_{rrtt}=- K  \frac{\rho^2}{m}  \left( \frac{p}{\rho} \right)^{\frac{\gamma}{2}+2}   \frac{ 2\sqrt{\pi}}{ \Gamma \! (\frac{4\alpha+\gamma+9}{2})}    \\
\times \left\{3(4\alpha +\gamma+4) \pi \Gamma \! \left(\alpha +\frac{\gamma}{2}+1\right)^2 +2^{\gamma+4} (\alpha+1) \Gamma \! \left(\alpha+1\right)^2 \Gamma\!  \left(\frac{\gamma+3}{2}\right)\Gamma \left(\frac{\gamma+5}{2}\right)\right\},
\end{multline*}
where the relation $p=\frac{\rho}{m} k T$ from \eqref{hydro p} was used. 

\subsubsection{Computation of $\sum_{r,t=1}^3 \mathcal P_{rtrt}$}
For the term
\begin{multline*}
 \sum_{r,t=1}^3 \mathcal P_{rtrt}=m  L_0^2 \int_{\mathbb{R}^6 \times [0,\infty)^2 \times [0,1]^2 \times S^2} e^{-\frac{1}{k T}\left( \frac{m}{2} (\vert c \vert ^2 + \vert c_* \vert^2)+I+I_*\right)} \\
 \times  \left( (c \cdot c')^2+(c \cdot c'_*)^2-(c \cdot c)^2-(c \cdot c_*)^2\right)\\
 \times \mathcal{B}^{nw} \phi_\alpha(r)  \,  (1-R) R^{\frac{1}{2}} \psi_\alpha(R)\, I^{\alpha}  I_*^{\alpha} \,  \mathrm{d} \sigma \, \mathrm{d} r \, \mathrm{d} R \,  \mathrm{d} I_* \, \mathrm{d}  c_*\, \mathrm{d} I \, \mathrm{d}  c.
\end{multline*}
After change of variable \eqref{c notation} term under integral yields
\begin{multline*}
(c \cdot c')^2+(c \cdot c'_*)^2-(c \cdot c)^2-(c \cdot c_*)^2=\\
 \frac{1}{2}(V_c  \cdot u')^2+\frac{1}{2}(V_c \cdot u')(V_c \cdot u)+\frac{1}{8}(u \cdot u')^2-\frac{1}{2}(V_c \cdot u)^2-\frac{1}{2}(V_c \cdot u)\vert u \vert^2-\frac{1}{8}\vert u \vert ^4.
\end{multline*}
Next, the form of cross section \eqref{model Bnw u} allow to perform an integration with respect to $V_c$, 
\begin{multline*}
 \sum_{r,t=1}^3 \mathcal P_{rtrt}=m K L_0^2  \left( \frac{\pi k T}{m} \right)^{\frac{3}{2}} \frac{1}{4}\int_{\mathbb{R}^3 \times [0,\infty)^2 \times [0,1]^2 \times S^2} e^{-\frac{1}{k T}\left(\frac{m}{4} \vert u \vert ^2+I+I_*\right)} \\
 \times \left(\frac{k T}{m} (\vert u' \vert ^2 - \vert u \vert)+\frac{1}{2} ((u' u)^2 -\vert u \vert^4)\right)\\
 \times \tilde{B} \phi_\alpha(r)  \,  (1-R) R^{\frac{1}{2}} \psi_\alpha(R)\, I^{\alpha}  I_*^{\alpha} \,  \mathrm{d} \sigma \, \mathrm{d} r \, \mathrm{d} R \,  \mathrm{d} I_*\, \mathrm{d} I \, \mathrm{d} u.
\end{multline*}
Using relations \eqref{R} and \eqref{I' I'_*}, we can express
\[ \vert u' \vert ^2=\frac{4 R E}{m}=R\vert u \vert^2+\frac{4R}{m} (I+I_*),\]
\[(u \cdot u')^2=\frac{4 R E}{m} (u \cdot \sigma)^2=\left( R\vert u \vert^2+\frac{4R}{m} (I+I_*) \right)(u \cdot \sigma)^2,\]
after which we perform integration respect to $\sigma$, that yields
\begin{multline*}
 \sum_{r,t=1}^3 \mathcal P_{rtrt}=m K L_0^2  \left( \frac{\pi k T}{m} \right)^{\frac{3}{2}}\pi \int_{\mathbb{R}^3 \times [0,\infty)^2 \times [0,1]^2} e^{-\frac{1}{k T}\left(\frac{m}{4} \vert u \vert ^2+I+I_*\right)} \\
 \times \left(\frac{kT}{m} \left(R-1\right)\vert u \vert^2+\frac{1}{2}\left(\frac{R}{3}-1\right)\vert u \vert^4+\frac{4R}{m}\left(I+I_*\right)\left(\frac{k T}{m}+\frac{1}{6} \vert u \vert ^2\right)\right)\\
 \times \tilde{B} \, \phi_\alpha(r)  \,  (1-R) R^{\frac{1}{2}} \psi_\alpha(R)\, I^{\alpha}  I_*^{\alpha} \, \mathrm{d} r \, \mathrm{d} R \,  \mathrm{d} I_*\, \mathrm{d} I \, \mathrm{d} u.
\end{multline*}
Now we switch to spherical coordinates for the  relative velocity $u$,  and integrate with respect to $r$ and $R$  using the notation \eqref{const r R} for the constants coming up from this integration,
\begin{multline*}
 \sum_{r,t=1}^3 \mathcal P_{rtrt}=m K L_0^2  \left( \frac{\pi k T}{m} \right)^{\frac{3}{2}}  4\pi^2 \int_{\times [0,\infty)^3 } e^{-\frac{1}{k T}\left(\frac{m}{4} y ^2+I+I_*\right)}y^2 I^{\alpha}  I_*^{\alpha} \\
 \times \left\{ -\frac{k T}{m} C_{\left(1,\frac{\gamma}{2},0\right)} y^{\gamma+2}+\frac{1}{2} \left( \frac{1}{3} C_{\left(0,\frac{\gamma}{2}+1,0\right)} - C_{\left(0,\frac{\gamma}{2},0\right)} \right) y^{\gamma+4}
 \right. \\
 \left.  +\frac{4}{m}C_{\left(0,1+\frac{\gamma}{2},0 \right)}\left(\frac{k T}{m}+\frac{1}{6}y^2 \right)(I+I_*)y^{\gamma}\right. \\
\left. +\left( - \frac{k T}{m}C_{\left(\frac{\gamma}{2}+1,0,\frac{\gamma}{2}\right)} y^2+\frac{1}{2} \left( \frac{1}{3} C_{\left(\frac{\gamma}{2},1,\frac{\gamma}{2}\right)} - C_{\left(\frac{\gamma}{2},0,\frac{\gamma}{2}\right)}\right) y^4 \right) \left(\left( \frac{I}{m}\right)^{\frac{\gamma}{2}}+ \left( \frac{I_*}{m}\right)^{\frac{\gamma}{2}}\right) \right. \\
 \left.+\frac{4}{m}\left(\frac{k T}{m}+\frac{1}{6}y^2\right) C_{\left(\frac{\gamma}{2},1,\frac{\gamma}{2}\right)}\left(\left( \frac{I}{m}\right)^{\frac{\gamma}{2}}+ \left( \frac{I_*}{m}\right)^{\frac{\gamma}{2}}\right)(I+I_*)\right\}
\,  \mathrm{d} I_*\, \mathrm{d} I \, \mathrm{d} y.
\end{multline*}
Finally, performing the integration with respect to $I$, $I_*$, and $y$ yields
\begin{multline*}
\sum_{r,t=1}^3\mathcal{P}_{rtrt}=- K \frac{\rho^2}{m}  \left( \frac{p}{\rho} \right)^{\frac{\gamma}{2}+2}   \frac{2 \sqrt{\pi}}{ \Gamma  (\frac{4\alpha+\gamma+9}{2})}   \left\{9(8\alpha +2\gamma+13) \pi \Gamma \left(\alpha +\frac{\gamma}{2}+1\right)^2 \right. \\
\left.+2^{\gamma+2} (4\alpha(\gamma+6)+\gamma(\gamma+12)+39) \Gamma \left(\alpha+1\right)^2 \Gamma \left(\frac{\gamma+3}{2}\right)\Gamma \left(\frac{\gamma+5}{2}\right)\right\}.
\end{multline*}

\subsection{Computation of $\overline Q_{i}^{14}$} 
The parity arguments imply that  $\mathcal Q_{in}$ vanishes unless $i=n$, which for the production term \eqref{Qbar} implies
\begin{equation*}
\overline Q_i ^{14}= \sum_{k=1}^3  U_k \overline P_{ki}^{14}+\left( \alpha +\frac{7}{2}\right)^{-1} \frac{\rho ^2}{mp^3} q_i \frac{1}{3} \sum_{r=1}^3 \mathcal Q_{rr}.
\end{equation*}

\subsubsection{Computation of $\sum_{r=1}^3 \mathcal Q_{rr}$}

We now compute the term
\begin{multline*}
 \sum_{r=1}^3 \mathcal Q_{rr}= L_0^2 \int_{\mathbb{R}^6 \times [0,\infty)^2 \times [0,1]^2 \times S^2} e^{-\frac{1}{k T}\left( \frac{m}{2} (\vert c \vert ^2 + \vert c_* \vert^2)+I+I_*\right)} \\
\left( \frac{m}{2}\vert c \vert^2 +I\right) \left( \left(\frac{m}{2} \vert c' \vert ^2+I'\right)c' \cdot c+\left(\frac{m}{2} \vert c'_{*} \vert ^2+I'_{*}\right)c'_* \cdot c \right.\\
\left. -\left(\frac{m}{2} \vert c \vert ^2+I\right)\vert c \vert^2-\left(\frac{m}{2} \vert c_{*} \vert ^2+I_{*}\right)c_{*} \cdot c\right)\\
 \times \mathcal{B}^{nw} \phi_\alpha(r)  \,  (1-R) R^{\frac{1}{2}} \psi_\alpha(R)\, I^{\alpha}  I_*^{\alpha} \,  \mathrm{d} \sigma \, \mathrm{d} r \, \mathrm{d} R \,  \mathrm{d} I_* \, \mathrm{d}  c_*\, \mathrm{d} I \, \mathrm{d}  c.
\end{multline*}
Switching to the center-of-mass framework by means of the  change of variables \eqref{c notation}, the term under integral becomes
\begin{multline*}
 \left( \left(\frac{m}{2} \vert c' \vert ^2+I'\right)c'\cdot c+\left(\frac{m}{2} \vert c'_{*} \vert ^2+I'_{*}\right)c'_* \cdot c \right.\\
\left. -\left(\frac{m}{2} \vert c \vert ^2+I\right)\vert c \vert^2-\left(\frac{m}{2} \vert c_{*} \vert ^2+I_{*}\right)c_{*}\cdot c\right)\\
= \frac{m}{2} (u' \cdot V_c)^2-\frac{m}{2}(u \cdot V_c)^2+\frac{m}{4}(u \cdot u')(u'   \cdot V_c)-\frac{m}{4}(u\cdot  V_c)\vert u \vert^2\\
+\frac{1}{2}(I'-I'_*)\left(u' \cdot V_c+\frac{1}{2} u \cdot  u' \right)-\frac{1}{2}(I-I_*)\left(u \cdot V_c+\frac{1}{2} \vert u \vert^2 \right).
\end{multline*}
The form of the cross section \eqref{model Bnw u}  allows to first integrate with  respect to $V_c$ and $\sigma$,
\begin{multline*}
 \sum_{r=1}^3 \mathcal Q_{rr}= K L_0^2 \left(\frac{\pi k T}{m} \right)^{\frac{3}{2}} 4 \pi \int_{\mathbb{R}^3 \times [0,\infty)^2 \times [0,1]^2} e^{-\frac{1}{k T}\left( \frac{m}{4} \vert u\vert ^2 +I+I_*\right)} \\
\left\{-\frac{1}{4}\left(\frac{5}{4} k T+\frac{m}{8} \vert u \vert ^2+I \right)(I-I_*) \vert u \vert ^2 \right. \\
\left.+\frac{m k T}{32} \vert u \vert ^2 \left( \left(\frac{5}{3} R-3 \right) \vert u \vert ^2+\frac{20R}{3m} (I+I_*) \right)\right. \\
\left. +\left(\frac{1}{4} I+\frac{5}{16} k T\right) k T\left( (R-1)\vert u \vert^2+\frac{4 R}{m}(I+I_*)\right)\right\}\\
 \times \tilde{B} \phi_\alpha(r)  \,  (1-R) R^{\frac{1}{2}} \psi_\alpha(R)\, I^{\alpha}  I_*^{\alpha} \, \mathrm{d} r \, \mathrm{d} R \,  \mathrm{d} I_* \, \mathrm{d} I \, \mathrm{d}  u.
\end{multline*}
Next, passing to the   spherical coordinates for the relative velocity $u$, denoting $\left|u\right|=y$, and integrating with respect to $R$ and $r$ we obtain
\begin{multline*}
 \sum_{r=1}^3 \mathcal Q_{rr}= K L_0^2 \left(\frac{\pi k T}{m} \right)^{\frac{3}{2}} 16 \pi^2 \int_{[0,\infty)^3} e^{-\frac{1}{k T}\left( \frac{m}{4} y^2 +I+I_*\right)} \\
y^2\left\{ y^{\gamma} \left[-\frac{1}{4}\left(\frac{5}{4} k T+\frac{m}{8} y ^2+I \right)(I-I_*)y ^2 C_{\left(0,\frac{\gamma}2,0\right)} \right. \right.\\
\left.\left. + \frac{m k T}{32}y^4 \left(\frac{5}{3}C_{\left(0,1+\frac{\gamma}{2},0\right)}-3C_{\left(0,\frac{\gamma}2,0\right)} \right)  -\left(\frac{1}{4} I+\frac{5}{16} k T\right) k T y^2 C_{\left(1,\frac{\gamma}2,0\right)}\right.\right. \\
\left.\left.+\frac{4}{m}(I+I_*)C_{\left(0,\frac{\gamma}{2}+1,0\right)} \left(\frac{5}{96} m k T y^2+\left(\frac{1}{4} I+\frac{5}{16} k T\right)k T\right)\right]\right. \\
\left. +\left(\frac{1}{m}\right)^{\frac{\gamma}{2}}\left( I^{\frac{\gamma}{2}}+I_*^{\frac{\gamma}{2}}\right)\left[-\frac{1}{4}\left(\frac{5}{4} k T+\frac{m}{8} y ^2+I \right)(I-I_*)y ^2 C_{\left(\frac{\gamma}2,0,\frac{\gamma}2\right)} \right. \right.\\
\left. \left.+\frac{m k T}{32}y^4 \left(\frac{5}{3}C_{\left(\frac{\gamma}2,1,\frac{\gamma}2\right)}-3C_{\left(\frac{\gamma}2,0,\frac{\gamma}2\right)} \right)-\left(\frac{1}{4} I+\frac{5}{16} k T\right) k T y^2 C_{\left(\frac{\gamma}2+1,0,\frac{\gamma}2\right)}\right. \right.\\
\left. \left.+\frac{4}{m}(I+I_*)C_{\left(\frac{\gamma}2,1,\frac{\gamma}2\right)} \left(\frac{5}{96} m k T y^2+\left(\frac{1}{4} I+\frac{5}{16} k T\right)k T\right)\right]\right\}
 I^{\alpha}  I_*^{\alpha} \, \mathrm{d} I_* \, \mathrm{d} I \, \mathrm{d}  y,
\end{multline*}
where the constants are defined in \eqref{const r R}. 
Finally, performing integration with respect to $I$, $I_*$, $y$ yields
\begin{multline*}
 \sum_{r=1}^3 \mathcal{Q}_{rr}=-K\rho^2  \left(\frac{ p}{\rho}\right)^{\frac{\gamma}{2}+3} \frac{ \sqrt{\pi}}{24 \Gamma\!  (\frac{4\alpha+\gamma+9}{2})}\\
\times \left\{9((4\alpha+\gamma)(2(4\alpha+\gamma)+\gamma^2+38)+7\gamma^2+160)\pi \Gamma\!\left(\alpha+\frac{\gamma}{2} + 1 \right)^2 \right. \\
\left.+2^{\gamma+5}((4\alpha+\gamma)(3\alpha+\gamma)+57\alpha+15\gamma+60) \Gamma \left(\alpha+1\right)^2 \Gamma \!\left(\frac{\gamma+3}{2}\right)\Gamma \!\left(\frac{\gamma+5}{2}\right) \right\}.
\end{multline*}

\medskip
\medskip

\end{document}